\documentclass[twocolumn]{article}
\pdfoutput=1

\usepackage[T1]{fontenc}
\usepackage{amsmath}
\usepackage{amsfonts}       
\usepackage{amsthm}
\usepackage{algpseudocode}
\usepackage{algorithm}
\usepackage[hidelinks]{hyperref}       
\usepackage{booktabs}       
\usepackage{thm-restate}
\usepackage{graphicx}
\usepackage[round]{natbib}

\newcommand{\KDE}[2]{\mathrm{KDE}_{#1}(#2)}
\DeclareMathOperator{\E}{E}
\usepackage[inline]{enumitem}
\newlist{compactenumi}{enumerate}{5}%
\setlist[compactenumi]{topsep=0pt,itemsep=-1ex,partopsep=1ex,parsep=1ex,%
   label=(\roman*)}%

\newlist{compactenuma}{enumerate}{5}%
\setlist[compactenuma]{topsep=0pt,itemsep=-1ex,partopsep=1ex,parsep=1ex,%
   label=(\alph*)}%

\newtheorem{theorem}{Theorem}
\newtheorem{lemma}[theorem]{Lemma}
\newtheorem{corr}[theorem]{Corollary}
\theoremstyle{definition}
\newtheorem{definition}[theorem]{Definition}

\begin{document}

%

%

\title{DEANN: Speeding up Kernel-Density Estimation using Approximate Nearest Neighbor Search}

\author{ Matti Karppa$^{1,2}$
  \and
  Martin Aum\"uller$^{1}$
  \and
  Rasmus Pagh$^{2,3}$
}

\date{\today}

\maketitle

\begin{abstract}
  Kernel Density Estimation (KDE) is a nonparametric method for
  estimating the shape of a density function, given a set of samples from the
  distribution. 
  Recently, \emph{locality-sensitive hashing}, originally proposed as a tool for nearest neighbor search, has been shown to enable fast KDE data structures.
  However, these approaches do not take advantage of the many other advances that have been made in algorithms for nearest neighbor algorithms.
  We present an algorithm called Density Estimation from
  Approximate Nearest Neighbors (DEANN) where we apply Approximate Nearest
  Neighbor (ANN) algorithms as a \emph{black box} subroutine to compute an unbiased KDE.
  The idea is to find points that have a large contribution to the KDE using ANN, compute their
  contribution exactly, and approximate the remainder with Random
  Sampling (RS). We present a theoretical argument that supports the
  idea that an ANN subroutine can speed up the
  evaluation. Furthermore, we provide a C++ implementation with a
  Python interface that can make use of an arbitrary ANN
  implementation as a subroutine for kernel density estimation. We show
  empirically that our implementation outperforms state of the art
  implementations in all high dimensional datasets we considered, and
  matches the performance of RS in cases where the ANN yield no gains
  in performance.
\end{abstract}

\footnotetext[1]{IT University of Copenhagen}
\footnotetext[2]{BARC}
\footnotetext[3]{University of Copenhagen}

\section{Introduction}
\label{sect:introduction}

\emph{Kernel Density Estimation (KDE)} is a nonparametric method for
estimating the shape of a density function, given a sample from the
distribution. For a \emph{dataset}~$X\subseteq \mathbb R^d$ and a \emph{kernel
function}~$K_h : \mathbb R^d \times \mathbb R^d \to [0,1]$, the kernel
density estimate of the \emph{query vector}~$y\in\mathbb R^d$ is given by
\begin{equation}
  \label{eq:kdeintro}
  \KDE{X}{y} = \frac{1}{|X|}\sum_{x\in X} K_h(x,y) \, .
\end{equation}
A common choice for the kernel function is the \emph{Gaussian kernel}
\begin{equation}
  \label{eq:gaussiankernel}
  K_h(x,y) = \exp\left(-\frac{||x-y||_2^2}{2h^2}\right) \, ,
\end{equation}
where the constant $h>0$ is the \emph{bandwidth} parameter. In the
one-dimensional case, the KDE has a simple interpretation with this
kernel function: given a set of points, plot a Gaussian Probability
Density Function (PDF) centered at each point, and the KDE is the
density function we get by taking the average of all these PDFs at
each point.
The bandwidth is thus the variance parameter, controlling the width of
each bell curve. The KDE may thus be viewed as a generalization of the
histogram with soft bins, and is routinely used for smoothing with
libraries such as Seaborn.\footnote{\url{https://seaborn.pydata.org/},
  see particularly the function \texttt{kdeplot}.}


The Gaussian kernel is an example of a \emph{radially decreasing
  kernel}, that is, its value depends only on the \emph{distance}
between the two operands~$x$ and~$y$, and is monotonically decreasing,
exponentially so. This family includes, for example,~the \emph{exponential
  kernel}~$K_h=\exp\left(-\frac{||x-y||_2}{h}\right)$ and the
\emph{Laplacian kernel} $K_h=\exp\left(-\frac{||x-y||_1}{h}\right)$.
Other common kernels
include the Epanechnikov kernel, the rectangular (or tophat) kernel,
or the triangular (or linear) kernel~\citep[Chapter~3]{Silverman:1986}, see also~\cite[Section~2.8.2]{sklearn-ug:2021}.
Though our methods will apply to any radial kernel, we focus on the exponentially decreasing radial kernels.

The KDE is easily generalized into the multivariate case. The
bandwidth may also be generalized into a cross-dimensional matrix that
corresponds to the covariance matrix, but we restrict ourselves to
scalar constant bandwidth. For kernels dependent only on the distance
between points, the bandwidth parameter can be seen as a scaling
parameter for the distances, and in practical applications, the choice
of proper bandwidth is important to ensure that the KDE values are
meaningful, that they show essential features of the underlying
distribution without becoming overly smooth while at the same time
avoiding the introduction of sampling
artifacts~\citep{JonesMS:1996}. It is immediate from
Equation~\eqref{eq:gaussiankernel} that, if we let $h\to \infty$, the
contribution of each summand in Equation~\eqref{eq:kdeintro}
approaches 1; conversely, if we let $h\to 0$, only the nearest
neighbors have significant contribution to the sum.

The KDE has seen use in applications such as estimating gradient lines
of densities~\citep{Arias-CastroMP:2016} and outlier
detection~\citep{SchubertZK:2014}. In machine learning, KDE is used in
classification~\citep{GanB:2017}.

The problem with a na\"ive application of Equation~\eqref{eq:kdeintro}
to compute the KDE value is that the sum depends on \emph{all} points
in the dataset; that is, an individual query requires $\Omega(nd)$
operations. If the number of queries is large, this may be prohibitively
expensive. An immediate improvement over the na\"ive summation is to
use Random Sampling (RS): it can be shown that computing the KDE on a
subset of $m=O(\frac{1}{\varepsilon^{2}\tau})$ points, sampled uniformly at
random with or without repetition, yields an unbiased estimator that
provides a relative $(1+\varepsilon)$-approximation guarantee on KDE
values in the excess of $\tau$, with constant probability.%
\footnote{If $\mu \geq \tau$ is the KDE value, the estimate $E$
  produced by the algorithm
  satisfies $\max\{E/\mu, \mu/E\} \leq (1+\varepsilon)$ with constant probability.}
Despite
this simplicity, it has turned out to be difficult to improve on RS
asymptotically whilst preserving theoretical guarantees in high
dimensions~\citep{CharikarS:2017}.

\subsection{Our contribution}
\label{sect:ourcontribution}

In this paper: 
\begin{compactenumi}
  \item We introduce an algorithmic approach to speed up kernel density
    estimation using approximate nearest neighbor algorithms as a
    black box. We call our approach \emph{DEANN}, for Density Estimation from
    Approximate Nearest Neighbors, 
  \item We provide a theoretical justification for the unbiasedness, and for the correctness and viability of our approach on real-world data, and
  \item We report on an extensive experimental study that compares our implementation to previous state-of-the-art approaches.
\end{compactenumi}
All of our code is available
online\footnote{\url{https://github.com/mkarppa/deann} } under the MIT
license, inlcuding 
the experimental
pipeline.\footnote{\url{https://github.com/mkarppa/deann-experiments}} 
The code includes dataset generation 
and preprocessing as well as post-processing of results, allowing for reproducibility and 
serving as a starting point for future work.

In more detail, a central idea in the attempt to speed up the evaluation of KDE sums
of the form of Equation~\eqref{eq:kdeintro} is to split the sum into
near and far components, depending on the distance to the dataset points
from the query vector.
We then wish to compute the contributions of the near points
exactly, and approximate the contribution of the far away points. 
However, we cannot hope to retrieve the actual nearest neighbors of the query point in a high-dimensional space efficiently, 
so we resort to ANN and compute the exact contribution of an \emph{approximate nearest neighbors} set.
Combining ANN and random sampling na\"ively does not result in an unbiased estimator, but 
we show how to efficiently correct for this bias. 
In fact, we obtain an estimator that is unbiased \emph{regardless} of
the
quality of the ANN data structure.
Only the variance of the resulting estimator is affected by the quality of the nearest neighbors approximation.

In Section~\ref{sect:theory} we formally define the DEANN algorithm,
prove that it is an unbiased estimator of the KDE value, and provide
theoretical arguments that support the idea that (and when) nearest
neigbors can help in the estimation of KDE values.  In
Section~\ref{sect:implementation}, we discuss our actual C++
implementation with a Python interface that can utilize an arbitrary
ANN implementation as a black box, and show in
Section~\ref{sect:experiments} that the result performs well in a
practical experimental setting. Due to lack of space, we have
relegated some of the additional experiments into the appendix.

\textbf{Limitations.} While our work is very general, this
generality also manifests itself in that we have so far no
theoretically grounded way to choose the parameters except empirical
grid search of the parameter space. Also, we are dependent on the ANN
subroutine which means we cannot provide a theoretical runtime
analysis for the algorithm without knowing the internals of the
ANN algorithm.

\subsection{Related work}

\textbf{Kernel density estimation.}
Three independent lines of research can be identified based on
space-partitioning trees, data sparsification, and Locality-Sensitive
Hashing (LSH). Methods based on creating a tree structure for partitioning
the search space
include~\citep{GrayM:2000,GrayM:2003,LeeGM:2005,LeeG:2008,MorariuSRDD:2008,RamLMG:2009},
but these methods are prone to suffer from the curse of
dimensionality. An interesting development of this line of research is
ASKIT~\citep{MarchXB:2015} that is in some cases able to perform also 
with high dimensional data if the data exhibits suitable structure;
the authors provide an implementation as free software.

In particular, \cite{MarchXB:2015} also use the idea of splitting up the contributions of near and far points,
but compute the contribution of far points in a different way. 
They prune the KDE computation in a tree-based space partitioning by approximating the contributions to the KDE value 
during a sub-tree traversal.
To apply this pruning, they run a bottom-up phase in the tree construction.
For each node in the tree, they look at the nearest neighbor 
information \emph{among the nodes in the sub-tree} and enrich these results with random samples. 
From that, they can store a short summary in the node.
This allows them to prune the computation at intermediate nodes in the top-down traversal for points 
that are guaranteed to be far away from the query.
In contrast to the approach of \cite{MarchXB:2015}, 
we use simple, data-independent random sampling, which is not only faster but also has the benefit of providing an unbiased estimator.

A second line of research includes \emph{$\varepsilon$-samples} or
\emph{coresets}~\citep{Phillips:2013,ZhengJPL:2013,PhillipsT:2020},
subsamples of the data that offer approximation guarantees.
Optimal coresets are often constructed as random samples with high-probability guarantees, and thus offer performance similar to RS.

The third line of work was initiated with the Hashing
Based Estimators (HBE) of~\cite{CharikarS:2017}. They applied importance
sampling to model KDE values through the collision probability of
Euclidean Locality Sensitive Hashing
(ELSH)~\citep{DatarIIM:2004}. Follow-up work includes Hashing Based
Sketches (HBS)~\citep{SiminelakisRBCL:2019} that was empirically shown
to outperform ASKIT, and the work of~\citet{BackursIW:2019} who presented an improvement
on the space usage. 
Very recently,~\cite{CharikarKNS20} further improved the asymptotic running time and space complexity in this line of research by using data-dependent LSH~\citep{AndoniLRW:2017}.

A more detailed discussion of the different methods 
is presented in Appendix~\ref{app:relatedwork}.

\textbf{Approximate Nearest Neighbor Search.}
Nearest neighbor search is a key primitive in many data mining and machine learning applications.
If vectors are embedded in a high-dimensional space, as is standard in computer vision~\citep{NetzerWCBWN:2011} or natural language processing~\citep{PenningtonSM:2014}, \emph{exact} nearest neighbor search becomes difficult,
an instance of the curse of dimensionality.

A long line of research focused on providing efficient implementations to find \emph{approximate} nearest neighbors. 
While these approaches often lack theoretical guarantees, they provide
a large speed-up over an exact linear scan with only a small loss in accuracy on real-world data;
see for example the large-scale evaluation study in~\cite{AumullerBF20}.
Several techniques can be used to build efficient ANN systems:
graph-based approaches, such as~\cite{Iwasaki18} and \cite{MalkovY20}, provide fast
query times but are expensive in preprocessing; cluster-based
techniques like~\cite{JohnsonDJ:2017} and \cite{GuoSLGSCK20} feature faster index
building times with a small loss in throughput. 
LSH-based approaches such as~\cite{AndoniILRS15} and \cite{AumullerCPV19} give
theoretical, probabilistic guarantees on the result quality, but are
often slower than the aforementioned approaches in practice. 


\section{Preliminaries}
\label{sect:preliminaries}

We write $[n] = \{0, 1, \ldots, n-1\}$. We say that a bijection $\pi\colon
[n] \to [n]$ is a \emph{permutation}.

We define the KDE problem formally as follows.
\begin{definition}[Kernel Density Estimate]
  Given a dataset $X=\{x_0,x_1,\ldots,x_{n-1}\}\subseteq\mathbb R^d$ of
  $d$-dimensional vectors, a constant bandwidth $h>0$, a kernel
  function $K_h : \mathbb R^d\times\mathbb R^d \to \mathbb R$, and a
  query vector $y\in\mathbb R^d$, we say that the Kernel Density
  Estimate (KDE) of $y$ is
  \[
    \mathrm{KDE}_X(y) = \frac 1n \sum_{i=0}^{n-1}K_h(x_i,y) \, .
  \]
  We often write $\mu = \mathrm{KDE}_X(y)$ when $y$, $X$, $h$, and
  $K_h$ are clear from the context.
\end{definition}

We call kernels that are monotonically decreasing functions of the
distance between a pair of points radially decreasing. If the
kernel $K_h$ is a function of the Euclidean distance of the pair of
points, such as the Gaussian or exponential kernels, we say it 
is \emph{Euclidean}.

Given the dataset $X\subseteq\mathbb{R}^d$ and a query vector $y\in\mathbb{R}^d$,
we denote with $(x'_0, \ldots, x'_{n-1})$ the sequence of dataset vectors sorted 
by distance to $y$.

We say that a random variable $Z$ is an unbiased estimator of $\mu$ if $E[Z]=\mu$. We present
the following well-known result that the KDE can be efficiently
approximated with random sampling. 
\begin{restatable}[Random Sampling]{lemma}{rslemma}
  \label{lem:rs}
  Let $X \subseteq \mathbb{R}^d, y \in \mathbb{R}^d$.
  Let $\tau \in (0,1)$ such that  $\KDE{X}{y} \geq \tau$.
  Drawing a uniform random sample $X'\subseteq X$
  (with repetition) of size $m=O(\frac{1}{\varepsilon^2\tau})$
  and computing $\KDE{X'}{y}$ yields an unbiased
  $(1+\varepsilon)$-approximation of $\KDE{X}{y}$, with constant probability.
\end{restatable}
\begin{proof}
  See Appendix~\ref{app:rsproof}.
\end{proof}
The bound on $m$ in Lemma~\ref{lem:rs} is tight up to a constant
for worst-case input (see Appendix~\ref{app:rsproof}).

\section{Algorithmic approach and theoretical foundations}
\label{sect:theory}

\subsection{Decomposing the KDE}
We start by proving the following lemma that states that 
the KDE of a query $y$ can be estimated from individual estimates 
on a partition of the dataset.

\begin{lemma}
  \label{lem:composition}
  Let the $n$-vector dataset $X\subseteq \mathbb R^d$ be partitioned
  into two non-empty parts $A,B\subseteq \mathbb{R}^d$, that is, $X=A\cup B$
  and $A \cap B = \emptyset$. Let $y\in\mathbb{R}^d$ be an arbitrary
  query vector, and let $Z_A$ and $Z_B$ be unbiased estimators of
  $\mathrm{KDE}_A(y)$ and $\mathrm{KDE}_B(y)$, respectively. Then,
  \[
    Z' = \frac{|A|}{n}Z_A + \frac{|B|}{n}Z_B
  \]
  is an unbiased estimator for $\mathrm{KDE}_X(y)$.
\end{lemma}

\begin{proof}
  By linearity of expectation and the definition of unbiased
  estimators, we have
  \begin{align*}
    \E[Z'] & = \E\left[\frac{|A|}{n}Z_A + \frac{|B|}{n}Z_B\right]
           = \frac{|A|}{n}\E[Z_A] + \frac{|B|}{n}\E[Z_B] \\
          & = \frac{|A|}{n}\frac{1}{|A|}\sum_{a\in A} K_h(a,y) +
            \frac{|B|}{n}\frac{1}{|B|}\sum_{b\in B}K_h(b,y) \\
          & = \frac{1}{n}\sum_{x\in A\cup B} K_h(x,y) 
           = \mathrm{KDE}_X(y) \enspace .
  \end{align*}
  
\end{proof}

\subsection{Algorithmic approach}

Given a query $y\in\mathbb R^d$ and a dataset
$X = \{x_0,x_1,\ldots,x_{n-1}\}\subseteq \mathbb R^d$ of $n$ points,
assume we have access to a black box subroutine $\mathrm{ANN}_X(y)$ that
returns (the indices of) $k$ approximate nearest neighbors $X_1 \subseteq X$ of
$y\in\mathbb R^d$. We can apply Algorithm~\ref{alg:annest} to compute
an unbiased estimate $\widetilde{\mathrm{KDE}}_X(y)$ of the KDE.

The algorithm works by partitioning the dataset into two parts: one
where all data points are close to the query vector, and the
remainder. The contribution of the near vectors is computed exactly,
and the remainder is approximated by random sampling. This idea bears
resemblance to that of the hierarchical tree methods, but is
expressed very concisely, and the nearest neighbors algorithm
is treated as black box. Indeed, the algorithm is very general: it admits
arbitrary kernels, metrics, and ANN algorithms, assuming
they are compatible with one another.

The algorithm has two parameters: the number of neighbors
to query~$k$
and the number of random samples~$m$. At the extremes, when
either $k$ or $m$ is zero, the algorithm either falls back to simple
random sampling, or simply discards all far points. Both cases may be
appropriate for certain datasets at very small or very large bandwidth
values. This also guarantees that the algorithm performs
asymptotically at least as well as simple random sampling.

Since $\textrm{KDE}_{X_1}(y)$ is the exact contribution of $k$ 
data points to the KDE of $y$, and a random sample on $X\setminus X_1$
results in an unbiased estimator of $\textrm{KDE}_{X\setminus X_1}(y)$,
we may conclude by Lemma~\ref{lem:composition} that Algorithm~\ref{alg:annest}
returns an unbiased estimator.

\begin{algorithm}[t]
  \caption{DEANN.}
  \label{alg:annest}
  \begin{tabular}{@{\hspace*{\algorithmicindent}}lp{0.72\linewidth}}
    \textbf{Input:} & Dataset~$X=\{x_0,x_1,\ldots,x_{n-1}\}\subseteq
                      \mathbb R^d$,
                      query vector~$y\in\mathbb R^d$,
                      kernel function $K_h\colon \mathbb R^d \times \mathbb
                      R^d \to \mathbb R$,
                      approximate nearest neighbor function
                      $\mathrm{ANN}_X\colon \mathbb R^d \to [n]^k$. \\
    \textbf{Output:} & Unbiased est.~$\widetilde{\mathrm{KDE}}_X(y)$ of $\KDE{X}{y}$.
  \end{tabular}
  \begin{algorithmic}[1]
    \Function{DEANN}{$X$, $K_h$, $\mathrm{ANN}_X$, $y$}
    \State $X_1 \gets \{x_i : i \in \mathrm{ANN}_X(y)\}$. \Comment{Find
      $k$ ANN}
    \State $X_2 \gets X\setminus X_1$. \Comment{$\{X_1,X_2\}$ is a
      partition of $X$.}
    \State $Z_1 \gets \KDE{X_1}{y} = \displaystyle \frac 1k\sum_{x\in X_1} K_h(x,y)$.
    \State $S \gets$ size-$m$ uniform random sample from~$X_2$. 
    \State $Z_2 \gets \KDE{S}{y} = \displaystyle \frac 1m \sum_{x \in
      S} K_h(x,y)$.
    \State $\widetilde{\mathrm{KDE}}_X(y) \gets \displaystyle \frac
    kn Z_1 + \frac{n-k}{n} Z_2$.
    \State \Return $\widetilde{\mathrm{KDE}}_X(y)$.
    \EndFunction
    \end{algorithmic}
  \end{algorithm}
  
\begin{corr}
  \label{cor:unbiased:estimate}
The value $\widetilde{\mathrm{KDE}}_X(y)$ in $\textnormal{DEANN}$ (Algorithm~\ref{alg:annest}) is an unbiased estimator of $\mathrm{KDE}_X(y)$.
\end{corr}
The estimate is unbiased no matter the quality of the near neighbors returned by $\mathrm{ANN}_X(y)$. 
This property is crucial: it allows us to use fast ANN implementations in practice that have no theoretical guarantees on the quality of their answers.



\subsection{Contribution of nearest neighbors in real-world datasets}

According to~\cite{DongWJCL08}, the distance distribution of distances from query points follows a Gamma distribution in many real-world datasets.
While the shape and scale parameters of the distribution may differ widely between various datasets, they can be estimated efficiently from a small sample.
As~\cite{DongWJCL08} observe, the same is true for the distance distribution of the $k$-th nearest neighbors.
In particular, \cite{PagelKF00} propose that the average distance of the $k$-th nearest neighbor under squared Euclidean distance can be modeled as a power-law function $\alpha (k/n)^\beta$, where $\alpha > 0$ is a constant depending on $d$, and $1/\beta > 1$ is the \emph{intrinsic dimensionality} of $X$. 

A rule of thumb for the selection of the bandwidth is to pick the median distance to the nearest neighbor as a bandwidth parameter~\citep{JaakkolaDH99}.
The following lemma shows that, given a distance distribution that follows a power-law distribution, this
bandwidth selection rule results in KDE values dominated by the contribution of a poly-logarithmic number of nearest neighbors. 
Deviating from this rule by much results in KDE values that
are \emph{meaningless}: too close to 0 or 1.

\begin{restatable}{lemma}{nncontributionlemma}
  \label{lem:nncontribution}
  Given $\alpha, 1/\beta > 0$, $X \subseteq \mathbb{R}^d$ with $|X| = n$, and $y \in \mathbb{R}^d$, assume that $||x'_i - y||_2^2 = \alpha ((i + 1)/n)^\beta$ for $i \in [n]$. 
  For the Gaussian kernel $K_h(x, y) = \exp\left(-||x -
    y||_2^2/(2h^2)\right)$,
  it holds that 
  \begin{compactenuma}
    \item If $h^2 = (\alpha/2) n^{-\beta}$, the contribution of the first $k=\Theta(\log^{1/\beta} n)$ nearest neighbors is a $(1+o(1))$-approximation of the KDE value.
    \item Let $\tau \in (0,1)$. If $h^2 \leq (\alpha/2) n^{-\beta}/\ln(1/\tau)$, $\textnormal{KDE}_X(y) \leq \tau$. 
    \item If $h^2 \geq \ln(1/(1 - \delta))\alpha/(2\beta)$, $\textnormal{KDE}_X(y) \geq 1 - \delta$.
  \end{compactenuma}
\end{restatable}

\begin{proof}
  See Appendix~\ref{app:nncontributionproof}.
\end{proof}

\subsection{How nearest neighbors help random sampling}

While the previous subsection gave a theoretical reason why the
rule-of-thumb for bandwidth selection is useful in practice, it
assumed exact distances and ignored the fact that, in practice,
$\log^{1/\beta} n$ might be a large number.
In general, every partition of the dataset $X$ into $S$ and $X\setminus S$ in Algorithm~\ref{alg:annest} results in an unbiased estimator. However,
it is unclear how the random sampling approach improves the estimate when 
the contribution of the $k$-nearest neighbors is known. 
This is because the number of samples $m$ in
Algorithm~\ref{alg:annest} is independent of the size $n - |S|$ of $X
\setminus S$ (see Lemma~\ref{lem:rs}).
The following definition and the resulting lemma show that the larger the 
contribution of the nearest neighbors, the fewer samples suffice to 
obtain a $(1 + \varepsilon)$-approximation of the KDE value.

\begin{definition}
  Given $n \geq 1$, $\delta \in (0, 1)$, and $k \in [n]$, let $X \subseteq \mathbb{R}^d$ with $|X| = n$. Given $y \in \mathbb{R}^d$, we say that 
  the pair $(k, \delta)$ \emph{dominates} $\textrm{KDE}_X(y)$ if 
    $\sum_{i = 0}^{k-1} K_h(x'_i, y) = (1-\delta) \sum_{i = 0}^{n-1} K_h(x'_i, y)$.
\end{definition}

The following lemma says that if the KDE value is $(k, \delta)$-dominated, 
a $\delta$-fraction of random samples is sufficient to obtain a $(1+\varepsilon)$-approximation.

\begin{restatable}{lemma}{dominateslemma}
  \label{lem:dominates}
  Let $\varepsilon > 0$, and $\textnormal{KDE}_X(y) \geq \tau$. If $(k, \delta)$ dominates
  $\textnormal{KDE}_X(y)$, then using $m = \Theta\left(\frac{\delta}{\varepsilon^2\tau}\right)$ samples guarantees that with constant probability, $\widetilde{\textnormal{KDE}}_X(y)$ is a $(1+\varepsilon)$-approximation. 
\end{restatable}

\begin{proof}
  See Appendix~\ref{app:dominatesproof}.
\end{proof}

\section{Implementation and engineering choices}
\label{sect:implementation}

\textbf{Implementation.}
We have implemented our algorithm in C++, using
Intel MKL as backend for linear algebra and
vectorized array computations. The implementation can be used as a
Python module, and accepts arbitrary ANN libraries
as a black box through a Python interface. 
We
provide example interfaces for using scikit-learn
\texttt{NearestNeighbors} as a baseline, and
FAISS~\citep{JohnsonDJ:2017} as a practical ANN implementation. The
code is available
online\footnote{\url{https://github.com/mkarppa/deann}}
under the MIT
license
and includes the na\"ive algorithm, random sampling, and DEANN.


\textbf{Optimizations for Euclidean kernels.}
While Algorithm~\ref{alg:annest} is agnostic to the
choice of the kernel, some further optimizations are possible if we
restrict ourselves to Euclidean kernels. 
We make the following observation regarding the Euclidean norm.
Using of the identity
$||x-y||_2^2=||x||_2^2+||y||_2^2-2\left<x,y\right>$ enables the
use of the matrix-matrix multiplication primitive \texttt{GEMM}
to speed-up batch
evaluation of Euclidean distances, described in more detail in
Appendix~\ref{app:naive}.

%

\textbf{Optimizing random sampling.}
A practical limitation of the random sampling routine is that a direct
implementation would mandate random access to memory.  To make
effective use of a CPU's prefetching ability, data must be accessed in a
linear or otherwise well-predictable fashion.
We speed up our random sampling scheme by
permuting the dataset vectors during preprocessing.
We can then take a contiguous subset of the
permuted vectors as the sample which can also be combined with the
matrix multiplication optimization described above, using the
matrix-vector multiplication primitive \texttt{GEMV}. For
completeness, pseudocode is given in Appendix~\ref{app:rsp}.
For a single query, this \emph{permuted random
  sampling} amounts to random sampling without replacement; however,
we lose independence when considering multiple queries. 
Although problematic when facing an adversary, the results are equally good
in practice, as shown in the next section. 

\section{Experiments}
\label{sect:experiments}


\textbf{Implementations.}
All implementations considered in our experiments are listed in
Table~\ref{tbl:implementations}.
We disambiguate implementations from abstract algorithms by writing
the name of the implementation in typewriter typeface. For example,
we distinguish between the naive and permuted random
sampling implementations by writing \texttt{RS} and \texttt{RSP},
respectively.
We refer to the variants of DEANN that use naive and permuted random
sampling as a subroutine by \texttt{DEANN} and \texttt{DEANNP}, respectively.
We evaluate our implementation against the \texttt{HBE} implementation of
\cite{SiminelakisRBCL:2019}, 
and the standard implementation
provided by scikit-learn~\citep{scikit-learn}.

The variant of \texttt{HBE} considered is called
\texttt{AdaptiveHBE} in the code of~\cite{SiminelakisRBCL:2019}, and
uses the HBS procedure~\cite[Algorithm~4]{SiminelakisRBCL:2019} for
subsampling the data and the Adaptive Mean Relaxation (AMR)
procedure~\cite[Algorithm~2]{SiminelakisRBCL:2019} for early
termination of queries. For completeness, we also evaluate the
\texttt{AdaptiveRS} variant of random sampling provided
by~\cite{SiminelakisRBCL:2019} that uses AMR with the RS
estimator, and denote it by \texttt{RSA}.
To our understanding, these are the particular varieties
evaluated in~\cite{SiminelakisRBCL:2019}. 
We instrumented their code to 
produce the output necessary in post-processing; the full version of 
their code used for this paper is accessible through the
\texttt{deann-experiments} repository.


We include the \texttt{KernelDensity} %
from scikit-learn~\citep{scikit-learn} as a baseline 
since scikit-learn is widely used in
practical data science applications. This implementation
uses $k$-d trees or ball trees with an optional error tolerance
parameter for accelerating KDE evaluations. We denote the two
different choices for data structure by \texttt{SKKD} and
\texttt{SKBT}, respectively.

We use FAISS~\citep{JohnsonDJ:2017} as the ANN implementation with our
estimator algorithms. 
In particular, we use their \emph{inverted file} index which runs $k$-means on the dataset. From the centroids of $k$-means, it builds a linear-space data structure in which each dataset point is assigned to its closest centroid. When answering a query, it inspects all points 
associated with the $n_q$ closest centroids to the query. 
Both $k$ and $n_q$ are user-defined parameters that are provided to the implementation.
Although FAISS supports extensive parallelism
with GPUs, we limit ourselves to the single-threaded CPU
version. This is because our implementation is entirely
single-threaded to make it comparable with pre-existing
single-threaded implementations; we also disabled multithreading in
MKL.

In the appendices, we provide additional evaluation results that
include (i) further considerations on the robustness of parameter choices
in Appendices~\ref{app:fullresults} and~\ref{app:fixedparameter}, and
(ii) experiments using the Gaussian kernel (including
ASKIT by~\cite{MarchXB:2015} as a competitor) in Appendix~\ref{app:gaussian}.
The trends observed in the main text translate well into these settings.

\begin{table}[t]
  \caption{Implementations Used in the Experiments.
    }
      \label{tbl:implementations}
      \centering
      \begin{tabular}{@{}l@{\,}p{0.35\linewidth}@{\,}p{0.5\linewidth}@{}}
        Name & Description                & Reference \\\hline
        \texttt{Naive} & Exact using \texttt{GEMM}  &   Section~\ref{sect:implementation} \\
        \texttt{RS} & Naive RS & Lemma~\ref{lem:rs} \\
        \texttt{RSP} & Permuted RS & Section~\ref{sect:implementation} \\
        \texttt{DEANN} & DEANN with \texttt{RS} & Section~\ref{sect:implementation} \\
        \texttt{DEANNP} & DEANN with \texttt{RSP} & Section~\ref{sect:implementation} \\
        \texttt{HBE} & HBE estimator & \cite{SiminelakisRBCL:2019} \\
        \texttt{RSA} & Adaptive RS &
                                                  \cite{SiminelakisRBCL:2019} \\
        \texttt{SKKD} & sklearn $k$-d-tree & \cite{scikit-learn} \\
        \texttt{SKBT} & sklearn balltree & \cite{scikit-learn} \\
      \end{tabular}
      \caption{Description of the Datasets.
      }
      \label{tbl:datasets}
      \centering
      \begin{tabular}{@{}l@{\,}r@{\,}r@{\,\,}l@{}}
        Dataset & $n$     & $d$ & Reference \\\hline
        \textsc{ALOI}    &   108,000 & 128 & \cite{GeusebroekBS:2005} \\
        \textsc{Census}  & 2,458,285 & 68  & \cite{Census} \\
        \textsc{Covtype} &   581,012 & 54  & \cite{BlackardD:1999} \\
        \textsc{GloVe}   & 1,193,514 & 100 & \cite{PenningtonSM:2014} \\
        \textsc{last.fm} &   292,385 & 65  & \cite{Celma:2010} \\
        \textsc{MNIST}   &    60,000 & 784 & \cite{LecunBBH:1998} \\
        \textsc{MSD}     &   515,345 & 90  & \cite{Bertin-MahieuxEWL:2011} \\
        \textsc{Shuttle} &    58,000 & 9   & \cite{Shuttle} \\
        \textsc{SVHN}    &   531,131 & 3072 & \cite{NetzerWCBWN:2011} \\
      \end{tabular}
  \end{table}

\textbf{Datasets.}
The datasets considered are presented in
Table~\ref{tbl:datasets}. The names of datasets are written in small
caps.
The choice of datasets includes the ones that
were used in previous works~\citep{SiminelakisRBCL:2019,BackursIW:2019}
for the sake of reproducibility of results, and also present variation
in the quality of data, the size of the dataset, and the number of
dimensions.
In all cases, we split the datasets in three disjoint subsets: a
validation set of 500 vectors, a test set of 500 vectors, and a
training set consisting of the remainder of the data. The training set
is used as the set~$X$ against which the KDE values are computed. The
validation and the test set are used as queries. 

\textbf{Bandwidth selection.} Following the approach in~\cite{BackursIW:2019}, 
we chose four \emph{target KDE values} $10^{-2}$, $10^{-3}$,
$10^{-4}$, and $10^{-5}$ and applied binary search on the validation
set to
find a bandwidth parameter $h$ such that the \emph{median} exact KDE value
of the validation set vectors is within a \emph{relative error}~\footnote{For
an individual query vector $y$, let the estimated KDE be $Z$ and the
correct KDE be $\mu$. We then say that the relative error is
$|Z-\mu|/\mu$. For a query set $Q=\{q_1,q_2,\ldots,q_m\}$ such that
the estimated KDE for the query vector $q_j$ is $Z_j$ and the correct
KDE is $\mu_j$, we say that the average relative error is
$\frac{1}{m}\sum_{j=1}^m |Z_j-\mu_j|/\mu_j$.} of 0.01 from
the target value.
The reason for this choice of multiple bandwidth values is that the
KDE values are very sensitive to a right choice of bandwidth; as the
bandwidth serves as a scaling factor to distances, a very
large bandwidth will make the distances meaningless and it does not
matter which points we look at, whereas a very small bandwidth 
together with the exponential decay of the kernel as a function of
distance means that the nearest neighbors completely determine the KDE
values. By trying different bandwidths, we explore the
intermediate region where both far-away points and nearby points
contribute to the typical density values.
For brevity, we will sometimes refer
to the target value by the letter $\mu$ in the remainder of this section.

\textbf{Experimental pipeline.}
We evaluate the validation set using the exponential 
kernel on
different algorithms and with different parameter values. The
supplementary material includes additional experiments with the
Gaussian kernel. The
parameters were chosen by a grid search over pre-selected parameter ranges;
see the supplemental code for detailed hyperparameter ranges. We exclude the parameter choices that exceed relative error 0.1, and then choose the fastest set of parameters with
respect to average query time.



The best choice of parameters is used to evaluate the test set, on
which we report the relative error, average query time, and the number
of samples looked at, as an average of five independent repetitions.
For \texttt{HBE}, we treat the relative
approximation error~$\varepsilon$ and the minimum KDE value~$\tau$ as
free parameters to be optimized. For the scikit-learn-based
implementations \texttt{SKKD} and \texttt{SKBT}, the parameters are
relative tolerance~$t_r$ which controls which subtrees the
implementation disregards, and the leaf size~$\ell$ of the evaluation
tree, where the implementation falls back to brute force. For
DEANN, the parameters are the number of nearest neighbors~$k$, 
the number of random samples to consider~$m$, the number of
clusters FAISS constructs~$n_\ell$, and the number of clusters FAISS
queries~$n_q$. 


\textbf{Machine details.}
The experiments were run on a shared computer with two 14-core Intel
Xeon E5-2690 v4 CPUs, amounting to 28 physical CPU cores, running at
2.6~GHz, 512~GiB RAM, and using Ubuntu 16.04 LTS. The code was
compiled with CLang~8.0.0, against Intel MKL version~2020.2, and the
experiments were run using CPython~3.8.5, NumPy~1.19.2,
scikit-learn~0.23.2, and FAISS version~1.7.0. The Python environment,
inlcuding MKL and FAISS, were managed through Anaconda~2020.11. A
small amount of other load was present on the computer.



\textbf{Results on validation set.}
Computing the KDE value with different methods on the validation set
provided the following insights: For target KDE values of $10^{-2}$
and $10^{-3}$, DEANN will usually fall back to random sampling which
provides faster query times.  For smaller KDE values, the best query
times were achieved by combining the contribution of the nearest
neighbors and random sampling.  Notable exceptions were
\textsc{last.fm} where using $k$ nearest neighbors pays off even for
large KDE values, and \textsc{GloVe} and \textsc{SVHN}, where random
sampling was the best choice for all $\mu$.

Table~\ref{tbl:parametres} lists the parameters that achieved the best
query time with respect to the validation set at relative error below
0.1 for a subset of datasets.
For lack of space, only the parameters for \texttt{RSP},
\texttt{DEANNP}, \texttt{HBE}, and \texttt{SKKD} are reported; the
parameters for other algorithms are very similar.
The subset was chosen to represent three different cases: a mixed case
(\textsc{ALOI}) where \texttt{DEANNP} performs the best for some
bandwidth choices and is on par with \texttt{RSP} for others, a case
that favors \texttt{DEANNP} (\textsc{last.fm}), and a case where
\texttt{RSP} performs the best (\textsc{SVHN}) and \texttt{DEANNP}
essentially falls back to random sampling. The full set of parameters
is reported in Appendix~\ref{app:fullresults}.

Table~\ref{tbl:recalls} shows the average recall rates for FAISS at
the choice of parameter that provided the best results. The subset of
results is different from Table~\ref{tbl:parametres} to highlight the
extrema. 
The
average fraction of true neighbors returned ranged from 0.23
(\textsc{ALOI}, $k=400$, $\mu=10^{-3}$) to 0.98 (\textsc{Shuttle},
$k=50$, $\mu=10^{-5}$) with a wide range of different values attained
between these extrema.
The full set of results together with an extended discussion is presented in 
Appendix~\ref{app:fullresults}.


\begin{table*}[h!!]
  \caption{The Best Choice Of Parameters Achieving Less Than 0.1
    Relative Error For A Subset Of Dataset/Target~$\mu$ Choices.}
  \label{tbl:parametres}
  \centering
\begin{tabular}{@{}l@{\enspace}l@{\enspace}r@{\enspace}r@{\enspace}r@{\enspace}r@{\enspace}r@{\enspace}r@{\enspace}r@{\enspace}r@{\enspace}r@{\enspace}r@{}}
     &  &  & \texttt{RSP} & \multicolumn{4}{c}{\texttt{DEANNP}} & \multicolumn{2}{c}{\texttt{HBE}} & \multicolumn{2}{c}{\texttt{SKKD}}\\
    \cmidrule(lr){4-4} \cmidrule(lr){5-8} \cmidrule(lr){9-10} \cmidrule(lr){11-12}
    Dataset & Target $\mu$ & $h$ & $m$ & $k$ & $m$ & $n_\ell$ & $n_q$ & $\epsilon$ & $\tau$ & $\ell$ & $t_r$\\\hline
\textsc{ALOI} &  0.01 & 3.3366 & 230 & 0 & 170 & 512 & 1 & 1.1 & 0.001 & 40 & 0.2\\
\textsc{ALOI} &  0.001 & 2.0346 & 1800 & 0 & 2100 & 512 & 1 & 0.6 & 0.0001 & 90 & 0.2\\
\textsc{ALOI} &  0.0001 & 1.3300 & 29000 & 170 & 500 & 1024 & 5 & \textit{n/a} & \textit{n/a} & 80 & 0.2\\
\textsc{ALOI} &  0.00001 & 0.8648 & 78000 & 120 & 430 & 1024 & 5 & \textit{n/a} & \textit{n/a} & 90 & 0.2\\
\textsc{last.fm} & 0.01 & 0.0041 & 75000 & 60 & 350 & 1024 & 1 & \textit{n/a} & \textit{n/a} & 10 & 0.2\\
\textsc{last.fm} & 0.001 & 0.0026 & 85000 & 70 & 800 & 512 & 1 & \textit{n/a} & \textit{n/a} & 10 & 0.15\\
\textsc{last.fm} & 0.0001 & 0.0019 & 160000 & 50 & 350 & 2048 & 5 & \textit{n/a} & \textit{n/a} & 20 & 0.1\\
\textsc{last.fm} & 0.00001 & 0.0015 & 200000 & 80 & 450 & 2048 & 5 & \textit{n/a} & \textit{n/a} & 100 & 0.15\\
\textsc{SVHN} & 0.01 & 632.7492 & 150 & 0 & 120 & 512 & 1 & 1.2 & 0.0001 & 70 & 0.2\\
\textsc{SVHN} & 0.001 & 391.3900 & 400 & 0 & 350 & 512 & 1 & \textit{n/a} & \textit{n/a} & 60 & 0.2\\
\textsc{SVHN} & 0.0001 & 277.1836 & 900 & 0 & 800 & 512 & 1 & \textit{n/a} & \textit{n/a} & 60 & 0.2\\
\textsc{SVHN} & 0.00001 & 211.4066 & 1900 & 0 & 2000 & 512 & 1 & \textit{n/a} & \textit{n/a} & 60 & 0.2\\
  \end{tabular}
\end{table*}
\begin{table*}[h!!]
  \caption{ Average Recall Rates For The Approximate
    Nearest Neighbors Returned By FAISS. }
  \label{tbl:recalls}
  \centering
  \begin{tabular}{llrrrrrrrrrr}
     &  & \multicolumn{5}{c}{\texttt{DEANN}} & \multicolumn{5}{c}{\texttt{DEANNP}}\\
    \cmidrule(lr){3-7} \cmidrule(lr){8-12}
    Dataset & Target $\mu$ & $k$ & $m$ & $n_\ell$ & $n_q$ & $R$ & $k$ & $m$ & $n_\ell$ & $n_q$ & $R$\\\hline
\textsc{ALOI} &  0.001 & 170 & 400 & 512 & 1 & 0.23 & \textit{n/a} & \textit{n/a} & \textit{n/a} & \textit{n/a} & \textit{n/a}\\
\textsc{ALOI} &  0.0001 & 200 & 430 & 1024 & 5 & 0.72 & 170 & 500 & 1024 & 5 & 0.74\\
\textsc{last.fm} & 0.01 & 50 & 400 & 2048 & 1 & 0.24 & 60 & 350 & 1024 & 1 & 0.88\\
\textsc{last.fm} & 0.001 & 70 & 200 & 2048 & 5 & 0.86 & 70 & 800 & 512 & 1 & 0.97\\
\textsc{MSD} & 0.00001 & 210 & 1800 & 4096 & 10 & 0.45 & 210 & 2100 & 2048 & 5 & 0.43\\
\textsc{Shuttle} & 0.00001 & 50 & 0 & 512 & 5 & 0.98 & 50 & 0 & 512 & 5 & 0.98\\
  \end{tabular}
\end{table*}
\begin{table*}[h!!]
  \caption{Results of Evaluating the Different Algorithms Against the
    Test Set in Milliseconds / Query.}
  \label{tbl:results}
  \centering
  \begin{tabular}{@{}l@{\,\,}l@{\,\,}r@{\,\,}r@{\,\,}r@{\,\,}r@{\,\,}r@{\,\,}r@{\,\,}r@{\,\,}r@{\,\,}r@{}}
    Dataset & Target $\mu$ & \texttt{Naive} & \texttt{RS} & \texttt{RSP} & \texttt{DEANN} & \texttt{DEANNP} & \texttt{HBE} & \texttt{RSA} & \texttt{SKKD} & \texttt{SKBT}\\\hline
    \textsc{ALOI} & 0.01 & 1.051 & 0.050 & 0.022 & 0.025 & \textbf{0.016} & 0.623 & 0.808 & 58.498 & 48.353\\
    \textsc{ALOI} & 0.001 & 1.058 & 0.326 & \textbf{0.105} & 0.211 & 0.148 & 12.192 & 41.411 & 59.353 & 47.644\\
    \textsc{ALOI} & 0.0001 & 1.055 & 6.477 & 1.698 & 0.270 & \textbf{0.197} & \textit{n/a} & \textit{n/a} & 55.786 & 47.916\\
    \textsc{ALOI} & 0.00001 & 1.057 & 21.781 & 4.548 & 0.219 & \textbf{0.182} & \textit{n/a} & \textit{n/a} & 47.930 & 49.698\\
    \textsc{last.fm} & 0.01 & 2.593 & 12.704 & 2.145 & 0.227 & \textbf{0.181} & \textit{n/a} & \textit{n/a} & 104.039 & 94.147\\
    \textsc{last.fm} & 0.001 & 2.621 & 17.183 & 2.455 & 0.277 & \textbf{0.222} & \textit{n/a} & \textit{n/a} & 99.893 & 86.006\\
    \textsc{last.fm} & 0.0001 & 2.753 & 48.630 & 4.699 & 0.294 & \textbf{0.247} & \textit{n/a} & \textit{n/a} & 98.582 & 83.999\\
    \textsc{last.fm} & 0.00001 & 2.923 & 40.249 & 5.993 & 0.330 & \textbf{0.263} & \textit{n/a} & \textit{n/a} & 85.621 & 83.367\\
    \textsc{SVHN} & 0.01 & 42.094 & 0.290 & \textbf{0.189} & 0.255 & 0.448 & 11.830 & 56.613 & 3447.218 & 2521.555\\
    \textsc{SVHN} & 0.001 & 42.172 & 0.747 & \textbf{0.500} & 0.698 & 0.938 & \textit{n/a} & 56.270 & 3471.669 & 2509.883\\
    \textsc{SVHN} & 0.0001 & 42.260 & 2.207 & \textbf{1.096} & 1.503 & 1.459 & \textit{n/a} & 83210.996 & 3455.433 & 2495.796\\
    \textsc{SVHN} & 0.00001 & 41.748 & 3.743 & \textbf{2.262} & 3.758 & 2.852 & \textit{n/a} & \textit{n/a} & 3496.380 & 2445.718\\
  \end{tabular}
\end{table*}

\textbf{Results on test set.} 
A subset of the main results are reported in
Table~\ref{tbl:results}, the same subset as in Table~\ref{tbl:parametres}. The full set of results is presented in
Appendix~\ref{app:fullresults}.
The table
lists the average query time per query vector in milliseconds, ordered
by the dataset and the target~$\mu$.

\begin{table*}[h!!]
  \begin{center}
    \caption{A Subset Of Preprocessing Times In Seconds.}
    \label{tbl:buildtimes}
\begin{tabular}{@{}l@{\,}l@{\,}r@{\,}r@{\,}r@{\,}r@{\,}r@{\,}r@{\,}r@{\,}r@{\,}r@{\,}r@{}}
Dataset & Target $\mu$ & \texttt{Naive} & \texttt{RS} & \texttt{RSP} & \texttt{DEANN} & \texttt{DEANNP} & \texttt{HBE} & \texttt{RSA} & \texttt{SKKD} & \texttt{SKBT} & \texttt{ASKIT}\\\hline\textsc{ALOI} & 0.01 & 0.006 & \textbf{0.000} & 0.055 & 0.377 & 8.775 & 22.285 & 0.000 & 4.929 & 5.155 & 21.455\\
\textsc{ALOI} & 0.00001 & \textbf{0.006} & \textit{n/a} & \textit{n/a} & 0.154 & 0.146 & \textit{n/a} & \textit{n/a} & 5.782 & 4.949 & 6.372\\
\textsc{Census} & 0.01 & 0.081 & \textbf{0.000} & 0.945 & 3.568 & 14.269 & 101.727 & 0.000 & 25573.250 & 22917.678 & \textit{n/a}\\
\textsc{Covtype} & 0.01 & 0.017 & \textbf{0.000} & 0.179 & 26.056 & 0.336 & 11.008 & 0.000 & 5.644 & 4.098 & 572.824\\
\textsc{Covtype} & 0.00001 & \textbf{0.016} & \textit{n/a} & \textit{n/a} & 0.593 & 0.621 & \textit{n/a} & \textit{n/a} & 5.026 & 4.010 & 75.267\\
\textsc{MNIST} & 0.01 & 0.017 & \textbf{0.000} & 0.159 & 1.700 & 0.813 & 100.323 & 0.000 & 12.369 & 11.154 & 14.053\\
\textsc{MNIST} & 0.00001 & 0.016 & \textbf{0.000} & 0.155 & 0.447 & 0.443 & \textit{n/a} & \textit{n/a} & 12.461 & 11.022 & 4.397\\
\textsc{MSD} & 0.01 & 0.020 & \textbf{0.000} & 0.223 & 9.319 & 9.395 & \textit{n/a} & \textit{n/a} & 12.378 & 10.359 & 144.805\\
\textsc{MSD} & 0.00001 & 0.019 & \textbf{0.000} & 0.224 & 0.446 & 0.460 & \textit{n/a} & \textit{n/a} & 11.614 & 10.028 & 144.940\\
\textsc{Shuttle} & 0.01 & 0.001 & \textbf{0.000} & 0.007 & 0.238 & 0.070 & 2.006 & 0.000 & 0.687 & 0.658 & 0.593\\
\textsc{SVHN} & 0.01 & 0.583 & \textbf{0.000} & 5.590 & 262.613 & 1651.727 & \textit{n/a} & \textit{n/a} & 454.764 & 473.117 & \textit{n/a}\\
\textsc{SVHN} & 0.00001 & 0.772 & \textbf{0.000} & 5.592 & 16.252 & 16.640 & \textit{n/a} & \textit{n/a} & 431.374 & 452.096 & \textit{n/a}\\
\end{tabular}
  \end{center}
\end{table*}

\textbf{Performance discussion.} 
In almost all cases, either \texttt{DEANNP}
or \texttt{RSP} was the fastest implementation, as indicated by bold
typeface (with the exception of \texttt{Covtype} at $\mu=10^{-5}$). 
In cases where \texttt{RSP} was the fastest algorithm, \texttt{DEANNP} does not
lose significantly because it falls back to random sampling; the
runtimes are very similar in those cases, apart from the slight
overhead of the more complex implementation.  
\texttt{RSP} provides speedups of a
factor of 2--10 for most workloads compared to 
\texttt{RS}.  In the small bandwidth regime where the ANN
contribution helps most, \texttt{RSP} is often slower by a factor of
10 or more than DEANN.  Contrasting our implementations to
competitors, we can compare to \texttt{HBE} consistently only for
target KDE value of $0.01$ and, usually, $0.001$.  In this setting,
performance is closest on \textsc{Covtype} with target KDE value
$0.001$ (\texttt{HBE} is roughly 2.5 times slower), but we observe a
speedup of 1-2 orders of magnitudes in many other settings, while being robust
even for very small target values.  The tree-based methods of
scikit-learn did not perform very well in our experiments. This is
largely due to the fact that the datasets are high-dimensional and the
space-partitioning methods tend to scale exponentially with
dimension. Indeed, scikit-learn performed
adequately in comparison to our \texttt{Naive} implementation only on 
\textsc{Shuttle}, the dataset with smallest~$d$,
and---surprisingly---\textsc{Covtype} with smallest target KDE.


\textbf{Task difficulty.} 
Some results are missing: for
\textsc{Shuttle} at target value of 0.00001, \texttt{RS} would have
required more samples than there are datapoints to achieve the desired
relative error. Several \texttt{HBE} and \texttt{RSA} results are 
missing due to our experimental setup, as a very
small value of $\tau$ ought to have been used to achieve a
sufficiently small relative error, as we included \emph{all} query
vectors in our experiments, even those with extremely small KDE
values. However, the implementation did not permit use of sufficiently
small $\tau$ values because either the runtimes grew excessively
large 
or the size of the data structure grew
so large that we ran out of RAM on our computer.  For finished runs, our results are in line with the results
in~\cite{SiminelakisRBCL:2019}.

\textbf{Preprocessing times.}
Our algorithm has no intrinsic data structure to
construct; the preprocessing time is determined by the 
ANN algorithm, and 
the time it takes
to create a permuted copy of the data for permuted sampling.
Table~\ref{tbl:buildtimes} shows a subset of preprocessing times that
have been collected when evaluating a similar set of experiments
against the Gaussian kernel. As such, this table also includes
\texttt{ASKIT} for comparison. The data points have been cherry-picked
to reflect various extreme cases, including the extreme case of over 7
hours for scikit-learn when constructing the tree for the
\textsc{Census} dataset. For \texttt{DEANN} and
\texttt{DEANNP}, the wide variation in the construction times is
determined by the choice of the FAISS parameters which provide a
tradeoff between construction and query time. Full results and discussion are presented in Appendix~\ref{app:preprocessingtimes}.

\textbf{Robustness considerations.}
Table~\ref{tbl:relerrs} shows the relative errors achieved when
evaluating the query set against the test set with the best
parameters, 
showing that DEANN
generalizes nicely: our
experiments show that this choice translated to a low average relative error
also in the test set, as the greatest individual observed value was on
\textsc{last.fm} at $\mu=0.01$ where the relative
error reached 0.114. The full set of results is presented in Appendix~\ref{app:fullresults}.

In Appendix~\ref{app:fixedparameter}, we discuss robust parameter selection for DEANN.
Instead of an expensive grid search, we report on experiments using one fixed set of parameters for different datasets and different target values. This single fixed parameter setting  provided
low relative error and good performance in most cases.

\subsubsection*{Acknowledgements}
  We thank Kexin Rong and Paris Siminelakis for helpful discussion
  regarding their code. 
  Matti Karppa and Rasmus Pagh are part of BARC, supported by VILLUM
  Foundation grant 16582.

\bibliographystyle{plainnat}
\bibliography{refs}


\clearpage
\appendix

\thispagestyle{empty}




\section{Asymptotic notation}

We use the asymptotic notation as defined by~\citet[Section~1.2.11]{Knuth:1997}. For $f,g : \mathbb N \to
\mathbb N$, we write $f(n) = O(g(n))$ if there exist positive
constants $n_0$ and $M$ such that $f(n) \leq M g(n)$ for all $n\geq
n_0$.
We also write $f(n) = \Omega(g(n))$ if there exist positive constants
$n_0$ and $L$ such that $f(n) \geq Lg(n)$ for all $n\geq n_0$. We
write $f(n) = \Theta(g(n))$ if $f(n) = O(g(n))$ and
$f(n)=\Omega(g(n))$.

Finally, for real-valued functions $f:\mathbb
N \to\mathbb R$, we write $f(n)=o(1)$ if $\lim_{n\to\infty} |f(n)| = 0$.

\section{Related work and historical perspectives on KDE}
\label{app:relatedwork}
This section provides an extended discussion on the related work, and
especially the historical discussion on earlier work.

Early developments in nontrivial computation of the KDE in low
dimensions include methods based on the Fast Fourier Transform, such
as~\cite{Silverman:1982} and~\cite{JonesL:1983,JonesL:1984} for the univariate
KDE, the Fast Multipole Method~\citep{GreengardR:1987}, and the Fast
Gauss Transform~\citep{GreengardS:1991}. This line of work has been
followed by a line of \emph{dual-tree} data
structures~\citep{GrayM:2000,GrayM:2003,LeeGM:2005,RamLMG:2009}. However,
these methods suffer from the curse of dimensionality. An attempt to
mitigate this effect in higher dimensions with \emph{subspace trees},
applying dimension reduction technologies such as Principal Component
Analysis (PCA) together with random sampling, was presented
by~\cite{LeeG:2008}, but even this method requires
$\Theta(\frac{1}{\epsilon^{2}})$ samples.

\cite{MorariuSRDD:2008} presented an algorithm based on tree data structures
and \emph{Improved Fast Gauss Transform} along with an
implementation called FigTree. \cite{MarchXB:2015} presented
\emph{ASKIT},
a tree-based space-partitioning method
based on \emph{treecodes} that can make efficient use of the low-rank
block structure of the matrix of pairwise kernel evaluations of the
data points even in high dimensions when such structure exists. They
also provided an implementation of ASKIT as free
software.\footnote{Available at~\url{https://padas.oden.utexas.edu/libaskit/}.}

Another line of research is focused on finding subsamples of the data
set that preserve the KDE values with arbitrary queries up to an
approximation factor, called \emph{$\epsilon$-samples} or
\emph{coresets}~\citep{Phillips:2013,ZhengJPL:2013,PhillipsT:2020}. However,
despite offering better approximation guarantees, asymptotically
coresets require a similar $\Theta(\frac{1}{\epsilon^2})$ number of
samples as simple Random Sampling.

There are also other approaches to subsampling the dataset, such as
Kernel Herding~\citep{ChenWS:2010}, and also
HBS~\citep{SiminelakisRBCL:2019} and the independent subsampling of
hash tables in~\citep{BackursIW:2019}.

\cite{CharikarS:2017} applied importance
sampling to model the KDE values through the collision probability of the
Euclidean Locality Sensitive Hashing (ELSH) scheme of \cite{DatarIIM:2004} to create a data
structure called \emph{Hashing Based Estimators (HBE)}. This data
structure presented first asymptotical improvement with theoretical
guarantees over simple RS in high dimensions. In
particular, HBE improves upon RS in the regime where a large amount of
the contribution comes from a small number of dataset points close to the
query point.

The theoretical nature of the results of~\cite{CharikarS:2017} were
made more practical by~\cite{SiminelakisRBCL:2019} who presented a data structure using
\emph{Hashing Based Sketches (HBS)}. Roughly, the idea of their KDE
estimation algorithm is to first subsample the dataset into a number
of sketches using ELSH and weighted sampling, and then construct the
HBE estimators from these subsampled datasets by reapplying ELSH, thus
``rehashing'' the dataset. They also presented an adaptive variant of
the algorithm whereby the ELSH data structures are constructed at a
number of levels, each containing an increasing number of hash tables,
corresponding to a lower bound of the estimated KDE value. Assuming a
sufficiently large KDE estimate can be made, the query terminates
early, but otherwise continues to a larger number of hash tables. They
also provide an implementation of their algorithm as free
software\footnote{Available at
  \url{https://github.com/kexinrong/rehashing}.} that can be used for
comparison. They showed empirically in~\citep{SiminelakisRBCL:2019}
that their HBE implementation is competitive with ASKIT and in some
performs an order of magnitude better than ASKIT.

Another improvement on the HBE scheme was presented by~\cite{BackursIW:2019} who improved on the space usage of
the algorithm by observing that HBE tends to store the same points in
several hash tables. They showed that, for each hash table, it suffices to
include each point hashed to the table with a certain probability to
guarantee that the point is stored in approximately one hash table,
and the approximation guarantees of HBE are still sufficiently
preserved. They provided a Python implementation\footnote{Available at
\url{https://github.com/talwagner/efficient_kde/}.} and used the
number of kernel function evaluations as a proxy for the runtime in
their experiments.

In recent work, \cite{CharikarKNS20} provided asymptotic improvements in running time and space complexity by using data-dependent LSH.

\section{Proof of Lemma~\ref{lem:rs}}
\label{app:rsproof}

In this appendix, we present the proof of Lemma~\ref{lem:rs}. The
proof is presented for completeness only without any claim to
originality. While the result is well known, it seems to be difficult
to find a useful version of the proof in the literature.

We need the following form of the Chernoff bound in the proof.
\begin{lemma}[{Chernoff~\cite[Theorem~1.1, pp.~6--7]{DevdattP:2009}}]
  Let $X = \sum_{i=1}^n X_i$ where $X_i \in [0,1]$ are independently distributed
  random variables. Then, for $\epsilon > 0$,
  \begin{align}
    \Pr[X > (1+\epsilon)\E[X]] & \leq
                                 \exp\left(-\frac{\epsilon^2}{3}\E[X]\right)
    \, , \label{eq:chernoff1}\\
    \Pr[X < (1-\epsilon)\E[X]] & \leq
                                 \exp\left(-\frac{\epsilon^2}{2}\E[X]\right)
                                 \, . \label{eq:chernoff2}
  \end{align}
  \label{lem:chernoff}
\end{lemma}

We recall Lemma~\ref{lem:rs}. We bound the number of random samples
required using the Chernoff bound with respect to an arbitrary
constant probability $\delta$.

\rslemma*
\begin{proof}
  Fix constant $0 < \delta < 1$. Let $X'=(x'_1,x'_2,\ldots,x'_m)$ be
  the random sample such that each $x'_i$ is drawn from $X$
  independently and uniformly distributed at random with
  repetition. We treat each $x'_i$ as a random variable taking values
  from the set $X$ and hold the
  query vector $y$
  arbitrary but fixed.

  For all $i=1,2,\ldots,m$, define 
  $Z_i=K_h(x'_i,y)$ where $K_h : \mathbb R^d\times \mathbb
  R^d \to [0,1]$ is the kernel function; without loss of generality, we may
  assume all $Z_i$ satisfy $0 \leq Z_i \leq 1$ by dividing the value
  of the kernel function with an appropriate constant. Clearly,
  $\E[Z_i] = \frac{1}{n}\sum_{j=1}^n K_h(x_i,y) = \mu$, so each $Z_i$
  is an unbiased estimator for $\mu=\KDE{X}{y}$.

  Letting $Z=\sum_{i=1}^mZ_i$, we get by linearity of expectation that
  $\E[Z] = m\E[Z_i] = m\mu \geq m\tau$. From
  Equation~\eqref{eq:chernoff1}, we get
  \begin{equation}
    \label{eq:chernoffapplication1}
    \Pr[Z > (1+\epsilon)\mu] \leq
    \exp\left(-\frac{\epsilon^2}{3}m\mu\right) \leq
    \exp\left(-\frac{\epsilon^2}{3}m\tau\right) \, .
  \end{equation}
  If we let the probability on the right hand side of
  Equation~\eqref{eq:chernoffapplication1} be less than or equal to the
  constant $\delta$, we get
  \[
    -\frac{\epsilon^2}{3}m\tau \leq \ln \delta \, ,
  \]
  and solving for $m$,
  \begin{equation}
    \label{eq:rsmbound}
    m \geq \frac{3\ln\frac{1}{\delta}}{\epsilon^2\tau} \, ,
  \end{equation}
  and by the same argument, Equation~\eqref{eq:chernoff2} yields the
  same bound on $m$ up to constant, 
  so we can thus conclude that $m=O(\frac{1}{\varepsilon^2\tau})$
  samples suffice to bound the error to the desired range. 
\end{proof}
It should be noted that, although not present in the
statement of Lemma~\ref{lem:rs}, the number of random samples~$m$
depends on the constant~$\delta$ by a factor of $\ln\frac{1}{\delta}$.

Furthermore, Lemma~\ref{lem:rs} is tight up to a constant. To see why,
we must consider a worst-case input that consists of vectors such that
a 
$\tau$-fraction of the dataset has kernel value of 1 and the remainder
are (essentially) 0. The random sample can be modelled as a sum of Bernoulli
variables such that the kernel values are either 0 or 1 with
probability $\tau$, which yields the correct KDE in expectation.

This input has a geometric interpretation, where the query is situated
such with respect to the dataset that a significant fraction (a
$\tau$-fraction) of the dataset essentially coincides with the query
vector (possibly up to a negligible amount of additive noise), and the
remainder of the dataset resides infinitely far (with respect to the
exponential decay of the kernel). This is precisely the regime where
we are looking for a needle in the haystack and nearest neighbors
essentially determine the KDE value, but we need to look at a large fraction of the dataset at
random to be able to find the needle.

We
will show that, with such input, the Chernoff bound is tight up to a constant, which
implies that also the required size of the sample is tight up to a
constant. To show this, we need the following lemma that we have
restated in the notation presented here.

\begin{lemma}[{\citet[Theorem~2.1]{Slud:1977}}]
  \label{lem:slud}
  Let $0\leq \tau \leq \frac{1}{4}$ and $\varepsilon > 0$. Let
  $X=\sum_{i=1}^m X_i$ with $X_i\sim\mathrm{Bernoulli}(\tau)$. Then
  \[
    \Pr[X \geq (1+\varepsilon)m\tau] \geq 1-\Phi\left(
      \frac{\varepsilon \sqrt{m\tau}}{\sqrt{1-\tau}} \right)
    > 1 - \Phi(2\varepsilon\sqrt{m\tau})
    \, ,
  \]
  where $\Phi$ is the standard normal cumulative distribution function.
\end{lemma}

\begin{lemma}
  Lemma~\ref{lem:rs} is tight up to a constant for worst-case input.
\end{lemma}
\begin{proof}
  This proof is almost the same as given by~\citet{Mousavi} and
  is presented here for completeness without claim to originality.
  
  Let us denote random variables $X_i$ for $i=1,2,\ldots,m$ such that each
  $X_i\sim \mathrm{Bernoulli}(\tau)$, yielding the worst-case
  input, drawn independently and identically distributed.
  As before, $X=\sum_{i=1}^{m} X_i$, so $E[X]=m\tau$.  
  Let us approximate $X$ with the normal distribution
  using Lemma~\ref{lem:slud}. It is known~\citep[Equation~3.7.2]{PatelR:1982} that,
  for $x>0$, 
  \[
    1-\Phi \geq \frac{1-\sqrt{1-\exp(-x^2)}}{2} \, .
  \]
  Furthermore, by the fact that $1-\sqrt{x}\geq \frac{1-x}{2}$,
  we can approximate
  \begin{equation}
    \label{eq:chernofflowerbound}
    \begin{split}
      & \Pr[X\geq(1+\varepsilon)E[X]] \\
      \geq & 1-\Phi(2\varepsilon\sqrt{m\tau}) \\
      \geq & \frac{1-\sqrt{1-\exp(-\varepsilon^2m\tau)}}{2} \\
      \geq & \frac{
        \exp(-\varepsilon^2 m \tau)}{4} \, ,
      \end{split}
  \end{equation}
  and since Equation~\eqref{eq:chernofflowerbound} is of the same form
  as the Chernoff bounds of Lemma~\ref{lem:chernoff}, we can conclude
  by the same argument as in the proof of Lemma~\ref{lem:rs} that the
  bound is tight up to a constant for the worst-case input.
\end{proof}

\section{Proof of Lemma~\ref{lem:nncontribution}}
\label{app:nncontributionproof}

We recall Lemma~\ref{lem:nncontribution}.
\nncontributionlemma*

\begin{proof}
  With $h^2 = (\alpha/2) n^{-\beta}$ the kernel evaluates to $K_h(x'_i,y) = \exp(-(i + 1)^\beta)$. With $k = \Theta(\log^{1/\beta} n)$, we get that $K_h(x'_i, y) = \exp(-(i+1)^\beta) = o(1/n)$ for all $i \geq k$. 
  Thus $\textrm{KDE}_{(x'_{k}, \ldots, x'_{n-1})}(y) = n \, o(1/n) = o(1)$, which proves the first statement.

  For the second statement, observe that with $h^2 \geq (\alpha/2) n^{-\beta}/ \ln(1/\tau)$, already the nearest neighbor evaluates to $K_h(x'_0, y) = \exp(-1/\ln(1/\tau)) = \tau$. 
  Since all other data points contribute at most $\tau$, $\textrm{KDE}_X(y) \leq \tau$.

  Finally, by the inequality of arithmetic and geometric means we can lower bound the KDE value as follows:

  \begin{align*}
    & 1/n \sum_{i=0}^{n-1} \exp(-\alpha ((i+1)/n)^\beta (1/h^2)) \\
    \geq & \prod_{i=0}^{n-1} \exp\left(-\alpha (i+1)^\beta n^{-\beta-1}(1/h^2)\right)\\
    = & \exp\left(-(\alpha/(h^2n^{\beta+1})) \sum_{i=1}^n i^\beta\right)\\
    \geq & \exp(-(\alpha/(h^2\beta))) \geq 1 - \delta \, . 
  \end{align*}
  Here, we used that $\sum_{i = 1}^n i^\beta = \frac{n^{\beta +
      1}}{\beta + 1} + O(n^\beta)$ and thus, asymptotically for large
  enough $n$,
  \[\sum_{i = 1}^n i^\beta < n^{\beta + 1}/\beta \, . \]
\end{proof}

\section{Proof of Lemma~\ref{lem:dominates}}
\label{app:dominatesproof}

We recall Lemma~\ref{lem:dominates}.
\dominateslemma*

\begin{proof}
  Given $y$, let $X = (x'_0, \ldots, x'_{n-1})$ be ordered in increasing order by distance to $y$.
  Given $\varepsilon' > 0$ to be set later, let $(n-k) \text{RS}_{(x'_{k}, \ldots, x'_{n-1})}(y)$
  be the value of an $(1+\varepsilon')$ approximation of $(n-k)\textrm{KDE}_{{(x'_{k}, \ldots, x'_{n-1})}}(y)$.
  We compute:
  \begin{align*}
    &\sum_{i = 0}^{k-1} K_h(x'_i, y) + (n-k)\text{RS}_{(x'_{k}, \ldots,
    x'_{n-1})}(y) \\
    \leq &
    \sum_{i = 0}^{k-1} K_h(x'_i, y) + (1+\varepsilon')\sum_{i = k}^{n-1} K_h(x'_i, y)\\
    = &n \textrm{KDE}(y) + \varepsilon'\sum_{i = k}^{n-1} K_h(x'_i, y)\\
    = &n (\textrm{KDE}(y) + \varepsilon'\delta \textrm{KDE}(y)) \, .
  \end{align*}
  This means that to compute a $(1 + \varepsilon)$ approximation, 
  it suffices to compute a $(1 + \varepsilon') = (1 + \varepsilon/\delta)$ approximation on
  $(x'_{k}, \ldots, x'_{n-1})$. 
  Since $\textrm{KDE}_{{(x'_{k}, \ldots, x'_{n-1})}}(y) \geq \delta\tau$, a sample
  of $\Theta\left(\frac{\delta}{\varepsilon^2\tau}\right)$ elements suffices to guarantee 
  a $(1+\varepsilon')$ approximation with constant probability.
\end{proof}

\section{Na\"ive algorithm}
\label{app:naive}

In this section, we describe how matrix multiplication can be used to
speed up the evaluation of the nai\"ve KDE sum when the kernel is
Euclidean. We make no claims of originality, but simply present the
material here for completeness. In this section, we treat the dataset
$X$ as a row-major $n\times d$ matrix.

Suppose we are working in a batch processing case with a set of
$N$~queries $Q=\{q_0,q_1,\ldots,q_{N-1}\}$ which we similarly treat as a
row-major $N\times d$ matrix. We want to evaluate the $N$-element
result vector $z$ whose elements are given by
\begin{equation}
  \label{eq:naivezj}
  z_j = \frac{1}{n}\sum_{i=0}^{n-1} K_h(q_j,x_i) \, .
\end{equation}
Assuming $K_h$ is Euclidean, the evaluation of
Equation~\eqref{eq:naivezj} for all $j=0,1,\ldots,N-1$ can be
considered to consist of (i) evaluating the $N\times n$ matrix $D$
whose elements are given by
\begin{equation}
  \label{eq:matrixddef}
  D_{j,i} = ||q_j-x_i||_2 \, ,
\end{equation}
(ii) applying the (vectorized) functions, the composition of which
equals $K_h$, and (iii) computing the row-wise mean of the resulting
matrix.

Matrix multiplication helps in step (i) through the following observation:
\begin{equation}
  \label{eq:sqeuclideannorm}
  ||x-y||_2^2 = ||x||_2^2 + ||y||_2^2 - 2\left<x,y\right> \, .
\end{equation}
Let us write auxiliary matrices $X_{\mathrm{sq}}$ and $Q_{\mathrm{sq}}$ such that for all $i=0,1,\ldots,n-1$ and $j=0,1,\ldots,N-1$, we have
\begin{equation}
  \label{eq:xrm}
  (X_{\mathrm{sq}})_{j,i} = ||x_i||_2^2 \, ,
\end{equation}
and
\begin{equation}
  \label{eq:qrm}
  (Q_{\mathrm{sq}})_{j,i} = ||q_j||_2^2 \, .
\end{equation}
Importantly, from Equations~\eqref{eq:xrm} and~\ref{eq:qrm}, we have that
\begin{equation}
  \label{eq:xsqpqsq}
  (X_{\mathrm{sq}} + Q_{\mathrm{sq}})_{j,i} = ||q_j||^2_2 + ||x_i||_2^2 \, .
\end{equation}

Now consider the matrix product $QX^\top$. From the definition of the matrix product, it is immediate that
\begin{equation}
  \label{eq:qxtinnerproduct}
  (QX^\top)_{j,i} = \left<q_j,x_i\right> \, .
\end{equation}
If we then let $D^2 = X_{\mathrm{sq}} + Q_{\mathrm{sq}} - 2QX^\top$, we get from Equations~\eqref{eq:sqeuclideannorm}, \eqref{eq:xsqpqsq}, and~\eqref{eq:qxtinnerproduct} that
\begin{equation}
  \label{eq:d2}
  D^2_{j,i} = ||q_j||^2_2 + ||x_i||_2^j - 2\left<q_j,x_i\right> = ||x_i-q_j||_2^2\, .
\end{equation}

The key observation is that it is possible to use matrix
multiplication as a primitive for evaluating the inner product matrix in
Equation~\eqref{eq:d2}. Evaluating the values of the matrix $D$
directly from the definition of Equation~\eqref{eq:matrixddef} one
element at a time requires $\Theta(nNd)$ operations. However, matrix
multiplication is asymptotically faster. For $n=N=d$, the evaluation
goes down to $O(n^\omega)$ operations for
$\omega < 2.3728639$~\citep{LeGall:2014}. Assuming $n=N$ and
$d<n^\alpha$ for $\alpha > 0.31389$, the evaluation can be performed in
$n^{2+o(1)}$ operations~\citep{LeGallU:2018}. Although these
theoretical developments are impractical, significant gains can be
made over implementing the evaluation naively even with the elementary
matrix multiplication algorithm by using, for example, the BLAS Level
3 subroutine \texttt{GEMM}\footnote{Generalized Matrix Multiply, a
  BLAS~\citep{BlackfordDDDHHHKLPPRW:2002} Level 3 subroutine for
  computing the matrix multiplication operation 
  $C \gets \alpha A^\top B + \beta C$.  The Intel MKL provides a
  highly optimized implementation of this routine.}
that is an
aggressively optimized
primitive~\citep{KagstromLL:1998,LiDT:2009,ZhangWZ:2012,AbdelfattahHTD:2016,KimCL:2019,YanWC:2020}. Highly
tuned implementations of \texttt{GEMM}, such as the one provided by
the Intel
MKL, make efficient use of the
CPU features, such as vectorization and cache hierarchy, and provide a
considerable performance boost over simple implementations.

\section{Permuted Random Sampling}
\label{app:rsp}

We present here for completeness the subroutine we use for taking the
optimized random sample in case of Euclidean kernels. Preprocessing
and sampling are presented in Algorithm~\ref{alg:permutedsampling}. We
make no claim to originality, and simply present the algorithm here
for completeness.

\setcounter{algorithm}{1}

\begin{algorithm}[h]
  \caption{Permuted random sampling.}
  \label{alg:permutedsampling}
  \begin{tabular}{@{\hspace*{\algorithmicindent}}lp{0.85\textwidth}}
    \textbf{Input:} & Dataset $X=\{x_0,x_1,\ldots,x_{n-1}\}\subseteq \mathbb R^d$
  \end{tabular}
  \begin{algorithmic}[1]
    \Procedure{Preprocess}{$X$}
    \State Draw permutation $\pi$ on $n$ elements at random.
    \State $X'\gets \{x'_0,x'_1,\ldots,x'_{n-1}\}$ such that $x'_i = x_{\pi(i)}$.
    \State $\ell \gets 0$. \Comment{Running index.}
    \EndProcedure
  \end{algorithmic}
  \hrule
  \vspace{1pt}
  \begin{tabular}{@{\hspace*{\algorithmicindent}}lp{0.35\textwidth}}
    \textbf{Input:} & Query vector $y\in\mathbb{R}^d$, integer number of samples $1
    \leq m \leq n$\\
    \textbf{Output:} & A random sample estimate of $\KDE{X}{y}$.
  \end{tabular}
  \begin{algorithmic}[1]
    \Function{RandomSamplePermuted}{$y$, $m$}
    \State $Z \gets \displaystyle \sum_{i = \ell}^{\ell+m-1}
    K_h(x'_{i\mod n},y)$.
    \State $\ell \gets \ell + m  \mod n$.
    \State \Return $\displaystyle \frac{1}{m}Z$.
    \EndFunction
  \end{algorithmic}
\end{algorithm}

Importantly, if the kernel $K_h$ is Euclidean, the evaluation of the
sample on line 2 can be treated as follows. First, we have either one
or two contiguous, rectangular submatrices of the permuted data
matrix; the latter case occures when the row index~$i$ overflows. We
can then consider the evaluation to take place such that we evaluate
the Euclidean distance to all points in the sample, evaluate the
kernel individually on each distance, possibly using vectorized
operations, and finally compute the mean.

Assume now that $\ell+m < n$. Let $x_{\mathrm{sq}}\in\mathbb{R}^m$ be
a vector of the squared norms of the vectors in the sample, that is,
$(x_{\mathrm{sq}})_j = ||x'_{\ell+j\mod n}||_2^2$ for
$j=0,1,\ldots,m-1$. The elements of this vector can be precomputed
during preprocessing. Then, let $X''$ be the $m\times d$ matrix
consisting of the rows $x'_\ell,x'_{\ell+1},\ldots,x'_{\ell+m-1}$. The
vector of squared Euclidean norms can then be computed in terms of
matrix-vector multiplication as follows:
\[
  z = x_{\mathrm{sq}} + X''y + ||y||_2^2 \, ,
\]
where the last scalar addition is considered to be broadcast to all
elements in the output vector. The matrix-vector product $X''y$ can be
evaluated efficiently using the 
\texttt{GEMV}\footnote{Generalized Matrix Vector multiply,
  a BLAS~\citep{BlackfordDDDHHHKLPPRW:2002} Level 2 subroutine for
  computing the matrix vector multiplication and addition operation of
  $y \gets \alpha Ax + \beta y$. The Intel MKL provides a highly
  optimized implementation of this routine.} subroutine. Generalization to
arbitrary cases follows by performing the operation in two steps
whenever the running index~$i$ overflows the size of the data matrix,
and in all cases by applying the relevant vectorized operations for
evaluating the kernel value.

\section{Detailed discussion of experimental evaluation with the
  exponential kernel}
\label{app:fullresults}

\paragraph{Full results.} Table~\ref{tbl:resultsfull} shows the full
set of results, average query time as milliseconds / query, when
evaluating the query set against the test set, with different
algorithms. The best values are indicated with a bold typeface.

\paragraph{Results on validation set.} Results of the validation step
of the experiments are presented in Table~\ref{tbl:parametresfull}. The
table lists the instances by dataset and target median KDE
value~$\mu$, the bandwidth~$h$ selected for the particular instance by
binary search with respect to the validation set, and the best
performing parameters for different algorithms. The parameters include
the number of random samples~$m$ for Permuted Random Sampling
(\texttt{RSP}), the number of nearest neighbors~$k$, the number of
random samples~$m$, the number of clusters~$n_{\ell}$, and the number
of clusters queried~$n_q$ by our ANN estimator \texttt{DEANNP} when
using FAISS,
the
relative approximation~$\epsilon$ and minimum KDE value~$\tau$ of the
\texttt{HBE} implementation, and the tree leaf size~$\ell$ and
relative error tolerance~$t_r$ for one the scikit-learn algorithms
\texttt{SKKD}. Due to
lack of space, the parameters for other variants are not shown but
they are very similar to the ones shown here.
In some cases, particularly for
\texttt{HBE}, no suitable choice of parameters was found, which is
indicated in the table by the text \textit{n/a}.

The bandwidth values are very small in cases where the nearest neighbors
help a lot with the performance. Indeed, in some cases, such as
\textsc{last.fm}, the bandwidth is below 1, meaning that it actually
expands the distances between the vectors. In some cases, such as
\textsc{Shuttle} at target $\mu$ of 0.00001, the random samples
provide such a small contribution to the overall KDE value that the
best performing parameters for the DEANN use no random samples at
all. Conversely, in several cases, such as all instances of
\textsc{SVHN}, the best choice of parameters for the DEANN was to fall
back to random sampling.

\paragraph{ANN recall.} In most cases, the number of clusters in the FAISS data structure
was rather large in comparison to the size of the dataset, but only
very few clusters were queried. 
This means that only a small fraction of the dataset was inspected to find 
nearest neighbors. While this is good for the throughput of the ANN
estimator, it might result in far-away points being included as
nearest neighbors,
or some true neighbors being missed. 
Let $\mathrm{NN}_k(q)$ and
$\widetilde{\mathrm{NN}}_k(q)$ be the correct set of $k$ nearest neighbors
for the query vector $q$ and the set returned by FAISS, respectively,
and let the query set~$Q$ be the validation set. The average recall
\[
  R = \frac{1}{|Q|}\sum_{q\in Q} \frac{|\mathrm{NN}_k(q) \cap
    \widetilde{\mathrm{NN}}_k(q)|}{|\mathrm{NN}_k(q)|}
\]
is reported per dataset and target KDE value in
Table~\ref{tbl:recallsfull} for both \texttt{DEANN} and
\texttt{DEANNP}.
The table only includes instances where a non-zero number of
nearest neighbors was queried, that is, cases where DEANN fell back to
random sampling are excluded. The table shows that a surprisingly small recall
is sometimes sufficient to achieve a small relative error. This is
particularly true for datasets where the majority of the contribution
came from the random samples. The extreme cases are \textsc{ALOI} at
$\mu=0.001$ with $k=170$ and $m=400$ where a measly $R=0.23$ was
sufficient to achieve the desired relative error, and, at the other
end, \textsc{Shuttle} at $\mu=0.00001$ with $k=50$ and $m=0$ where we
got $R=0.98$.

\paragraph{Robustness considerations.}
Table~\ref{tbl:relerrsall} shows empirically that DEANN generalizes
nicely. The parameters were chosen such that the average relative
error did not exceed 0.1 in the validation set; the table shows that
this translates to low average relative error also in the test set. The
greatest individual observed value was on \textsc{last.fm} at
a target value of 0.01 where the average relative error reached 0.114.

Figure~\ref{fig:errvsqueries} shows the dependence between different
parameter choices from the validation step. Different parameter
choices are plotted and the corresponding average relative error is
shown on the $x$-axis and the effect on runtime---the number of
queries processed per second---on the $y$-axis. Each individual
parameter choice is presented with a marker, and to help visualize the
dependence, a lineplot is drawn between the markers. Each subplot
corresponds to a single dataset, and the different target KDE values
are shown in the same plot with different colors and markers. Only
meaningful parameter choices are shown here; parameter choices that
would yield a worse relative error without gain in query speed are
excluded. The figure shows that the parameter choices form a clear
tradeoff between approximation quality and runtime, meaning it is
possible to tune DEANN to various use cases, depending on the
requirements on approximation quality and query times.

\section{Fixed-parameter experiments}
\label{app:fixedparameter}

In this section, we report on experiments that we carried out using a
\emph{fixed set of parameters}, that is, we made an educated guess for
the constants $k$ and $m$, and evaluated each dataset / target~$\mu$
combination against this choice of parameters using the test set as
the query set. The point of this exercise is to show that the
expensive grid search is not necessary for a practical application;
that it is, in fact, possible to find good enough parameters by
evaluating a the algorithm against a small sample with a good guess of
parameters. This shows the robustness of our algorithm: that it is not
sensitive to the exact correct choice of parameters.

Table~\ref{tbl:fixparamresults} shows the results of evaluating
\texttt{DEANNP} against the test set with the exponential kernel using
the fixed parameters $k=100$, $m=1000$, $n_\ell=512$, and $n_q=1$. The
table lists the time per query, the average relative error, and the
corresponding runtime for the best parameters obtained from the grid
search for comparison at relative error below 0.1. As expected, a
fixed choice of parameters favors some dataset/bandwidth choices more
than others, but overall, the results are encouraging. In terms of
error, the worst behavior is observed in the case of \textsc{Census}
with small bandwidths, and the reason is clear: too few neighbors are
looked at; this is also reflected in the runtime which is more than a
factor of 2 faster than with the parameters that achieve the error
below 0.1. To the other extreme, in the case of \textsc{GloVe} with
the large bandwidth, we get a relative error of 0.014, suggesting that
we could have done with a lot fewer samples.

The practical implication of this exercise is that it suggests the
following procedure for a practical application of the algorithm:
Choose a smallish query set, make a guess of parameters, evaluate
against ground truth, and if the results are not good enough (too high
error or too high runtime), refine the parameters by taking a new guess; since the algorithm
behaves in a very predictable manner, only very few guesses should
suffice in a practical setting to find ``good enough'' parameters,
meaning that an expensive hyperparameter tuning may not always be necessary.

\section{Experiments with the Gaussian kernel}
\label{app:gaussian}

\begin{table*}[h!!]
  \caption{The Full Set of Results of Evaluating the Different Algorithms Against the
    Test Set in Milliseconds / Query.}
  \label{tbl:resultsfull}
  \centering
  \begin{tabular}{@{}l@{\,\,}l@{\,\,}r@{\,\,}r@{\,\,}r@{\,\,}r@{\,\,}r@{\,\,}r@{\,\,}r@{\,\,}r@{\,\,}r@{}}
    Dataset & Target $\mu$ & \texttt{Naive} & \texttt{RS} &
                                                            \texttt{RSP} & \texttt{DEANN} & \texttt{DEANNP} & \texttt{HBE} & \texttt{RSA} & \texttt{SKKD} & \texttt{SKBT}\\\hline
        \textsc{ALOI} & 0.01 & 1.051 & 0.050 & 0.022 & 0.025 & \textbf{0.016} & 0.623 & 0.808 & 58.498 & 48.353\\
    \textsc{ALOI} & 0.001 & 1.058 & 0.326 & \textbf{0.105} & 0.211 & 0.148 & 12.192 & 41.411 & 59.353 & 47.644\\
    \textsc{ALOI} & 0.0001 & 1.055 & 6.477 & 1.698 & 0.270 & \textbf{0.197} & \textit{n/a} & \textit{n/a} & 55.786 & 47.916\\
    \textsc{ALOI} & 0.00001 & 1.057 & 21.781 & 4.548 & 0.219 & \textbf{0.182} & \textit{n/a} & \textit{n/a} & 47.930 & 49.698\\
    \textsc{Census} & 0.01 & 21.201 & 0.257 & \textbf{0.045} & 0.185 & 0.082 & 0.705 & 19.493 & 420.866 & 542.229\\
    \textsc{Census} & 0.001 & 21.821 & 1.268 & \textbf{0.192} & 0.902 & 0.215 & \textit{n/a} & 803.509 & 350.470 & 606.949\\
    \textsc{Census} & 0.0001 & 51.656 & 8.648 & 1.723 & 1.237 & \textbf{0.757} & \textit{n/a} & \textit{n/a} & 253.440 & 462.727\\
    \textsc{Census} & 0.00001 & 22.282 & 51.162 & 9.037 & 1.312 & \textbf{0.736} & \textit{n/a} & \textit{n/a} & 207.266 & 366.852\\
    \textsc{Covtype} & 0.01 & 4.921 & 1.036 & 0.128 & 0.269 & \textbf{0.055} & 0.314 & 20.534 & 46.734 & 50.446\\
    \textsc{Covtype} & 0.001 & 4.913 & 1.797 & \textbf{0.222} & 0.678 & 0.279 & 0.629 & 433.858 & 26.425 & 28.755\\
    \textsc{Covtype} & 0.0001 & 5.992 & 8.182 & 1.824 & 0.596 & \textbf{0.473} & \textit{n/a} & \textit{n/a} & 11.348 & 13.923\\
    \textsc{Covtype} & 0.00001 & 7.818 & 94.322 & 10.177 & \textbf{0.223} & 0.265 & \textit{n/a} & \textit{n/a} & 3.953 & 6.098\\
    \textsc{GloVe} & 0.01 & 11.302 & 0.011 & \textbf{0.001} & 0.005 & 0.003 & 0.347 & 0.207 & 674.429 & 582.650\\
    \textsc{GloVe} & 0.001 & 11.054 & 0.019 & \textbf{0.003} & 0.012 & 0.007 & 6.617 & 0.225 & 699.529 & 586.988\\
    \textsc{GloVe} & 0.0001 & 11.050 & 0.030 & \textbf{0.005} & 0.019 & 0.014 & \textit{n/a} & 0.410 & 704.741 & 581.489\\
    \textsc{GloVe} & 0.00001 & 11.101 & 0.048 & \textbf{0.015} & 0.041 & 0.022 & \textit{n/a} & 1.804 & 709.414 & 621.037\\
    \textsc{last.fm} & 0.01 & 2.593 & 12.704 & 2.145 & 0.227 & \textbf{0.181} & \textit{n/a} & \textit{n/a} & 104.039 & 94.147\\
    \textsc{last.fm} & 0.001 & 2.621 & 17.183 & 2.455 & 0.277 & \textbf{0.222} & \textit{n/a} & \textit{n/a} & 99.893 & 86.006\\
    \textsc{last.fm} & 0.0001 & 2.753 & 48.630 & 4.699 & 0.294 & \textbf{0.247} & \textit{n/a} & \textit{n/a} & 98.582 & 83.999\\
    \textsc{last.fm} & 0.00001 & 2.923 & 40.249 & 5.993 & 0.330 & \textbf{0.263} & \textit{n/a} & \textit{n/a} & 85.621 & 83.367\\
    \textsc{MNIST} & 0.01 & 1.495 & 0.029 & \textbf{0.024} & 0.024 & 0.029 & 1.577 & 0.884 & 94.960 & 63.640\\
    \textsc{MNIST} & 0.001 & 1.507 & 0.090 & \textbf{0.062} & 0.091 & 0.065 & 12.073 & 6.886 & 94.545 & 61.830\\
    \textsc{MNIST} & 0.0001 & 1.504 & 0.422 & 0.213 & 0.345 & \textbf{0.202} & \textit{n/a} & 8.915 & 89.835 & 59.892\\
    \textsc{MNIST} & 0.00001 & 1.524 & 1.172 & 0.773 & 0.609 & \textbf{0.536} & \textit{n/a} & \textit{n/a} & 94.857 & 64.299\\
    \textsc{MSD} & 0.01 & 4.725 & 0.053 & \textbf{0.016} & 0.033 & 0.028 & \textit{n/a} & 1.196 & 181.871 & 209.109\\
    \textsc{MSD} & 0.001 & 4.720 & 0.196 & \textbf{0.065} & 0.248 & 0.066 & \textit{n/a} & 88.375 & 165.613 & 197.519\\
    \textsc{MSD} & 0.0001 & 4.729 & 1.301 & \textbf{0.234} & 0.461 & 0.266 & \textit{n/a} & \textit{n/a} & 171.721 & 203.407\\
    \textsc{MSD} & 0.00001 & 4.754 & 9.898 & 1.482 & 0.754 & \textbf{0.405} & \textit{n/a} & \textit{n/a} & 127.574 & 169.668\\
    \textsc{Shuttle} & 0.01 & 0.407 & 0.145 & \textbf{0.017} & 0.138 & 0.024 & 0.308 & 8.207 & 3.671 & 4.097\\
    \textsc{Shuttle} & 0.001 & 0.402 & 0.864 & \textbf{0.062} & 0.141 & 0.113 & 1.595 & 398.961 & 2.525 & 3.873\\
    \textsc{Shuttle} & 0.0001 & 0.569 & 3.088 & 0.358 & 0.113 & \textbf{0.097} & 545.129 & \textit{n/a} & 1.917 & 3.437\\
    \textsc{Shuttle} & 0.00001 & 0.672 & \textit{n/a} & 0.527 & 0.070 & \textbf{0.065} & \textit{n/a} & \textit{n/a} & 1.064 & 2.436\\
    \textsc{SVHN} & 0.01 & 42.094 & 0.290 & \textbf{0.189} & 0.255 & 0.448 & 11.830 & 56.613 & 3447.218 & 2521.555\\
    \textsc{SVHN} & 0.001 & 42.172 & 0.747 & \textbf{0.500} & 0.698 & 0.938 & \textit{n/a} & 56.270 & 3471.669 & 2509.883\\
    \textsc{SVHN} & 0.0001 & 42.260 & 2.207 & \textbf{1.096} & 1.503 & 1.459 & \textit{n/a} & 83210.996 & 3455.433 & 2495.796\\
    \textsc{SVHN} & 0.00001 & 41.748 & 3.743 & \textbf{2.262} & 3.758 & 2.852 & \textit{n/a} & \textit{n/a} & 3496.380 & 2445.718\\
  \end{tabular}
\end{table*}

\begin{table*}[h!!]
  \caption{Results Of The Validation Step Of The Experiments Including
    The Best Performing Parameters For Some Algorithms.}
  \label{tbl:parametresfull}
  \centering
\begin{tabular}{@{}l@{\enspace}l@{\enspace}r@{\enspace}r@{\enspace}r@{\enspace}r@{\enspace}r@{\enspace}r@{\enspace}r@{\enspace}r@{\enspace}r@{\enspace}r@{}}
     &  &  & \texttt{RSP} & \multicolumn{4}{c}{\texttt{DEANNP}} & \multicolumn{2}{c}{\texttt{HBE}} & \multicolumn{2}{c}{\texttt{SKKD}}\\
    \cmidrule(lr){4-4} \cmidrule(lr){5-8} \cmidrule(lr){9-10} \cmidrule(lr){11-12}
    Dataset & Target $\mu$ & $h$ & $m$ & $k$ & $m$ & $n_\ell$ & $n_q$ & $\epsilon$ & $\tau$ & $\ell$ & $t_r$\\\hline
\textsc{ALOI} &  0.01 & 3.3366 & 230 & 0 & 170 & 512 & 1 & 1.1 & 0.001 & 40 & 0.2\\
\textsc{ALOI} &  0.001 & 2.0346 & 1800 & 0 & 2100 & 512 & 1 & 0.6 & 0.0001 & 90 & 0.2\\
\textsc{ALOI} &  0.0001 & 1.3300 & 29000 & 170 & 500 & 1024 & 5 & \textit{n/a} & \textit{n/a} & 80 & 0.2\\
\textsc{ALOI} &  0.00001 & 0.8648 & 78000 & 120 & 430 & 1024 & 5 & \textit{n/a} & \textit{n/a} & 90 & 0.2\\
\textsc{Census} & 0.01 & 3.6228 & 1000 & 0 & 800 & 512 & 1 & 0.95 & 0.0005 & 80 & 0.4\\
\textsc{Census} & 0.001 & 1.9416 & 6000 & 0 & 5000 & 512 & 1 & \textit{n/a} & \textit{n/a} & 100 & 0.25\\
\textsc{Census} & 0.0001 & 1.1907 & 40000 & 700 & 5500 & 1024 & 1 & \textit{n/a} & \textit{n/a} & 10 & 0.2\\
\textsc{Census} & 0.00001 & 0.7826 & 300000 & 800 & 5000 & 4096 & 5 & \textit{n/a} & \textit{n/a} & 60 & 0.2\\
\textsc{Covtype} & 0.01 & 245.8858 & 5000 & 0 & 1300 & 512 & 1 & 1.3 & 0.0001 & 90 & 0.3\\
\textsc{Covtype} & 0.001 & 119.2450 & 9000 & 0 & 8500 & 512 & 1 & 1.5 & 0.0001 & 100 & 0.2\\
\textsc{Covtype} & 0.0001 & 63.4887 & 70000 & 1300 & 1400 & 2048 & 5 & \textit{n/a} & \textit{n/a} & 30 & 0.2\\
\textsc{Covtype} & 0.00001 & 33.1331 & 350000 & 300 & 500 & 2048 & 5 & \textit{n/a} & \textit{n/a} & 100 & 0.2\\
\textsc{GloVe} & 0.01 & 1.5782 & 20 & 0 & 20 & 512 & 1 & 1.2 & 0.001 & 90 & 0.15\\
\textsc{GloVe} & 0.001 & 1.0372 & 50 & 0 & 50 & 512 & 1 & 0.75 & 0.0001 & 50 & 0.2\\
\textsc{GloVe} & 0.0001 & 0.7674 & 90 & 0 & 90 & 512 & 1 & \textit{n/a} & \textit{n/a} & 50 & 0.1\\
\textsc{GloVe} & 0.00001 & 0.6028 & 160 & 0 & 160 & 512 & 1 & \textit{n/a} & \textit{n/a} & 90 & 0.2\\
\textsc{last.fm} & 0.01 & 0.0041 & 75000 & 60 & 350 & 1024 & 1 & \textit{n/a} & \textit{n/a} & 10 & 0.2\\
\textsc{last.fm} & 0.001 & 0.0026 & 85000 & 70 & 800 & 512 & 1 & \textit{n/a} & \textit{n/a} & 10 & 0.15\\
\textsc{last.fm} & 0.0001 & 0.0019 & 160000 & 50 & 350 & 2048 & 5 & \textit{n/a} & \textit{n/a} & 20 & 0.1\\
\textsc{last.fm} & 0.00001 & 0.0015 & 200000 & 80 & 450 & 2048 & 5 & \textit{n/a} & \textit{n/a} & 100 & 0.15\\
\textsc{MNIST} & 0.01 & 532.9814 & 40 & 0 & 40 & 512 & 1 & 1.2 & 0.001 & 50 & 0.2\\
\textsc{MNIST} & 0.001 & 348.4158 & 150 & 0 & 150 & 512 & 1 & 1.05 & 0.0001 & 50 & 0.0\\
\textsc{MNIST} & 0.0001 & 255.3234 & 600 & 0 & 600 & 512 & 1 & \textit{n/a} & \textit{n/a} & 100 & 0.5\\
\textsc{MNIST} & 0.00001 & 198.7733 & 2200 & 140 & 450 & 512 & 5 & \textit{n/a} & \textit{n/a} & 50 & 0.0\\
\textsc{MSD} & 0.01 & 498.4585 & 230 & 0 & 230 & 512 & 1 & \textit{n/a} & \textit{n/a} & 90 & 0.2\\
\textsc{MSD} & 0.001 & 312.7048 & 1200 & 0 & 1000 & 512 & 1 & \textit{n/a} & \textit{n/a} & 90 & 0.2\\
\textsc{MSD} & 0.0001 & 222.0082 & 5500 & 0 & 5300 & 512 & 1 & \textit{n/a} & \textit{n/a} & 90 & 0.1\\
\textsc{MSD} & 0.00001 & 168.9344 & 36000 & 210 & 2100 & 2048 & 5 & \textit{n/a} & \textit{n/a} & 20 & 0.2\\
\textsc{Shuttle} & 0.01 & 4.9727 & 1900 & 0 & 1900 & 512 & 1 & 1.1 & 0.0001 & 20 & 0.2\\
\textsc{Shuttle} & 0.001 & 2.3504 & 11000 & 200 & 500 & 512 & 5 & 1.0 & 0.00001 & 60 & 0.2\\
\textsc{Shuttle} & 0.0001 & 1.1605 & 45000 & 100 & 500 & 512 & 5 & 0.1 & 0.000005 & 100 & 0.2\\
\textsc{Shuttle} & 0.00001 & 0.5648 & 52000 & 50 & 0 & 512 & 5 & \textit{n/a} & \textit{n/a} & 10 & 0.2\\
\textsc{SVHN} & 0.01 & 632.7492 & 150 & 0 & 120 & 512 & 1 & 1.2 & 0.0001 & 70 & 0.2\\
\textsc{SVHN} & 0.001 & 391.3900 & 400 & 0 & 350 & 512 & 1 & \textit{n/a} & \textit{n/a} & 60 & 0.2\\
\textsc{SVHN} & 0.0001 & 277.1836 & 900 & 0 & 800 & 512 & 1 & \textit{n/a} & \textit{n/a} & 60 & 0.2\\
\textsc{SVHN} & 0.00001 & 211.4066 & 1900 & 0 & 2000 & 512 & 1 & \textit{n/a} & \textit{n/a} & 60 & 0.2\\
  \end{tabular}
\end{table*}

\begin{table*}[h!!]
\caption{Recall Rates For Approximate Nearest Neighbors Returned By
  FAISS At Different Parameter Values.}
  \label{tbl:recallsfull}
  \centering
  \begin{tabular}{llrrrrrrrrrr}
     &  & \multicolumn{5}{c}{\texttt{DEANN}} & \multicolumn{5}{c}{\texttt{DEANNP}}\\
    \cmidrule(lr){3-7} \cmidrule(lr){8-12}
    Dataset & Target $\mu$ & $k$ & $m$ & $n_\ell$ & $n_q$ & $R$ & $k$ & $m$ & $n_\ell$ & $n_q$ & $R$\\\hline
\textsc{ALOI} &  0.001 & 170 & 400 & 512 & 1 & 0.23 & \textit{n/a} & \textit{n/a} & \textit{n/a} & \textit{n/a} & \textit{n/a}\\
\textsc{ALOI} &  0.0001 & 200 & 430 & 1024 & 5 & 0.72 & 170 & 500 & 1024 & 5 & 0.74\\
\textsc{ALOI} &  0.00001 & 200 & 270 & 1024 & 5 & 0.72 & 120 & 430 & 1024 & 5 & 0.78\\
\textsc{Census} & 0.001 & 500 & 3000 & 4096 & 1 & 0.31 & \textit{n/a} & \textit{n/a} & \textit{n/a} & \textit{n/a} & \textit{n/a}\\
\textsc{Census} & 0.0001 & 1300 & 3000 & 4096 & 5 & 0.50 & 700 & 5500 & 1024 & 1 & 0.69\\
\textsc{Census} & 0.00001 & 1400 & 3000 & 4096 & 10 & 0.77 & 800 & 5000 & 4096 & 5 & 0.58\\
\textsc{Covtype} & 0.001 & 1200 & 1100 & 1024 & 5 & 0.97 & \textit{n/a} & \textit{n/a} & \textit{n/a} & \textit{n/a} & \textit{n/a}\\
\textsc{Covtype} & 0.0001 & 900 & 1000 & 1024 & 5 & 0.99 & 1300 & 1400 & 2048 & 5 & 0.72\\
\textsc{Covtype} & 0.00001 & 350 & 0 & 2048 & 5 & 0.85 & 300 & 500 & 2048 & 5 & 0.86\\
\textsc{last.fm} & 0.01 & 50 & 400 & 2048 & 1 & 0.24 & 60 & 350 & 1024 & 1 & 0.88\\
\textsc{last.fm} & 0.001 & 70 & 200 & 2048 & 5 & 0.86 & 70 & 800 & 512 & 1 & 0.97\\
\textsc{last.fm} & 0.0001 & 70 & 300 & 2048 & 5 & 0.86 & 50 & 350 & 2048 & 5 & 0.90\\
\textsc{last.fm} & 0.00001 & 80 & 400 & 2048 & 5 & 0.86 & 80 & 450 & 2048 & 5 & 0.86\\
\textsc{MNIST} & 0.00001 & 400 & 300 & 512 & 5 & 0.77 & 140 & 450 & 512 & 5 & 0.95\\
\textsc{MSD} & 0.0001 & 140 & 1000 & 2048 & 5 & 0.46 & \textit{n/a} & \textit{n/a} & \textit{n/a} & \textit{n/a} & \textit{n/a}\\
\textsc{MSD} & 0.00001 & 210 & 1800 & 4096 & 10 & 0.45 & 210 & 2100 & 2048 & 5 & 0.43\\
\textsc{Shuttle} & 0.001 & 300 & 350 & 512 & 5 & 0.84 & 200 & 500 & 512 & 5 & 0.87\\
\textsc{Shuttle} & 0.0001 & 200 & 200 & 512 & 5 & 0.87 & 100 & 500 & 512 & 5 & 0.89\\
\textsc{Shuttle} & 0.00001 & 50 & 0 & 512 & 5 & 0.98 & 50 & 0 & 512 & 5 & 0.98\\
  \end{tabular}
\end{table*}

\begin{table*}[h!!]
  \caption{Average Relative Error Against The Test Set With Best
    Parameters.}
  \label{tbl:relerrsall}
  \centering
  \begin{tabular}{@{}l@{\enspace}l@{\enspace}r@{\enspace}r@{\enspace}r@{\enspace}r@{\enspace}r@{\enspace}r@{\enspace}r@{\enspace}r@{\enspace}r@{}}
    Dataset & Target $\mu$ & \texttt{Naive} & \texttt{RS} & \texttt{RSP} & \texttt{DEANN} & \texttt{DEANNP} & \texttt{HBE} & \texttt{RSA} & \texttt{SKKD} & \texttt{SKBT}\\\hline
    \textsc{ALOI} & 0.01 & 0.000 & 0.095 & 0.090 & 0.100 & 0.102 & 0.110 & 0.099 & 0.076 & 0.091\\
    \textsc{ALOI} & 0.001 & 0.000 & 0.106 & 0.113 & 0.104 & 0.101 & 0.096 & 0.097 & 0.092 & 0.097\\
    \textsc{ALOI} & 0.0001 & 0.000 & 0.102 & 0.099 & 0.100 & 0.100 & \textit{n/a} & \textit{n/a} & 0.098 & 0.098\\
    \textsc{ALOI} & 0.00001 & 0.000 & 0.072 & 0.102 & 0.092 & 0.094 & \textit{n/a} & \textit{n/a} & 0.099 & 0.098\\
    \textsc{Census} & 0.01 & 0.001 & 0.081 & 0.087 & 0.087 & 0.094 & 0.090 & 0.079 & 0.092 & 0.087\\
    \textsc{Census} & 0.001 & 0.002 & 0.087 & 0.082 & 0.094 & 0.091 & \textit{n/a} & 0.064 & 0.091 & 0.095\\
    \textsc{Census} & 0.0001 & 0.002 & 0.084 & 0.088 & 0.103 & 0.105 & \textit{n/a} & \textit{n/a} & 0.088 & 0.099\\
    \textsc{Census} & 0.00001 & 0.001 & 0.077 & 0.079 & 0.095 & 0.103 & \textit{n/a} & \textit{n/a} & 0.094 & 0.098\\
    \textsc{Covtype} & 0.01 & 0.001 & 0.047 & 0.045 & 0.094 & 0.094 & 0.095 & 0.086 & 0.098 & 0.099\\
    \textsc{Covtype} & 0.001 & 0.000 & 0.093 & 0.094 & 0.098 & 0.097 & 0.099 & 0.065 & 0.081 & 0.088\\
    \textsc{Covtype} & 0.0001 & 0.000 & 0.142 & 0.097 & 0.096 & 0.092 & \textit{n/a} & \textit{n/a} & 0.090 & 0.090\\
    \textsc{Covtype} & 0.00001 & 0.000 & 0.074 & 0.098 & 0.098 & 0.093 & \textit{n/a} & \textit{n/a} & 0.092 & 0.087\\
    \textsc{GloVe} & 0.01 & 0.000 & 0.095 & 0.096 & 0.095 & 0.097 & 0.124 & 0.089 & 0.069 & 0.096\\
    \textsc{GloVe} & 0.001 & 0.000 & 0.093 & 0.092 & 0.093 & 0.093 & 0.091 & 0.090 & 0.097 & 0.070\\
    \textsc{GloVe} & 0.0001 & 0.000 & 0.095 & 0.095 & 0.102 & 0.098 & \textit{n/a} & 0.108 & 0.047 & 0.080\\
    \textsc{GloVe} & 0.00001 & 0.000 & 0.097 & 0.098 & 0.096 & 0.098 & \textit{n/a} & 0.060 & 0.090 & 0.020\\
    \textsc{last.fm} & 0.01 & 0.001 & 0.061 & 0.052 & 0.111 & 0.114 & \textit{n/a} & \textit{n/a} & 0.094 & 0.091\\
    \textsc{last.fm} & 0.001 & 0.001 & 0.095 & 0.092 & 0.111 & 0.089 & \textit{n/a} & \textit{n/a} & 0.086 & 0.056\\
    \textsc{last.fm} & 0.0001 & 0.002 & 0.056 & 0.086 & 0.109 & 0.108 & \textit{n/a} & \textit{n/a} & 0.051 & 0.073\\
    \textsc{last.fm} & 0.00001 & 0.004 & 0.093 & 0.088 & 0.092 & 0.096 & \textit{n/a} & \textit{n/a} & 0.105 & 0.161\\
    \textsc{MNIST} & 0.01 & 0.000 & 0.090 & 0.094 & 0.091 & 0.092 & 0.103 & 0.093 & 0.082 & 0.093\\
    \textsc{MNIST} & 0.001 & 0.000 & 0.098 & 0.097 & 0.094 & 0.096 & 0.093 & 0.083 & 0.000 & 0.000\\
    \textsc{MNIST} & 0.0001 & 0.000 & 0.088 & 0.095 & 0.092 & 0.093 & \textit{n/a} & 0.104 & 0.006 & 0.000\\
    \textsc{MNIST} & 0.00001 & 0.000 & 0.102 & 0.100 & 0.098 & 0.094 & \textit{n/a} & \textit{n/a} & 0.000 & 0.000\\
    \textsc{MSD} & 0.01 & 0.000 & 0.103 & 0.097 & 0.097 & 0.100 & \textit{n/a} & 0.068 & 0.080 & 0.087\\
    \textsc{MSD} & 0.001 & 0.000 & 0.101 & 0.148 & 0.091 & 0.107 & \textit{n/a} & 0.097 & 0.091 & 0.095\\
    \textsc{MSD} & 0.0001 & 0.000 & 0.148 & 0.096 & 0.107 & 0.098 & \textit{n/a} & \textit{n/a} & 0.047 & 0.098\\
    \textsc{MSD} & 0.00001 & 0.000 & 0.096 & 0.091 & 0.103 & 0.100 & \textit{n/a} & \textit{n/a} & 0.096 & 0.099\\
    \textsc{Shuttle} & 0.01 & 0.000 & 0.094 & 0.095 & 0.096 & 0.098 & 0.105 & 0.091 & 0.080 & 0.093\\
    \textsc{Shuttle} & 0.001 & 0.000 & 0.119 & 0.102 & 0.099 & 0.101 & 0.090 & 0.069 & 0.091 & 0.095\\
    \textsc{Shuttle} & 0.0001 & 0.002 & 0.120 & 0.065 & 0.096 & 0.102 & 0.097 & \textit{n/a} & 0.095 & 0.094\\
    \textsc{Shuttle} & 0.00001 & 0.002 & \textit{n/a} & 0.084 & 0.073 & 0.073 & \textit{n/a} & \textit{n/a} & 0.094 & 0.090\\
    \textsc{SVHN} & 0.01 & 0.000 & 0.081 & 0.081 & 0.092 & 0.093 & 0.109 & 0.048 & 0.098 & 0.098\\
    \textsc{SVHN} & 0.001 & 0.000 & 0.084 & 0.084 & 0.088 & 0.090 & \textit{n/a} & 0.080 & 0.099 & 0.099\\
    \textsc{SVHN} & 0.0001 & 0.000 & 0.076 & 0.087 & 0.090 & 0.091 & \textit{n/a} & 0.053 & 0.099 & 0.099\\
    \textsc{SVHN} & 0.00001 & 0.000 & 0.090 & 0.098 & 0.089 & 0.091 & \textit{n/a} & \textit{n/a} & 0.099 & 0.099\\
  \end{tabular}
\end{table*}

\begin{table*}[ht]
  \caption{
    Average Relative Error Against The Test Set With Best Parameters.
  }
  \label{tbl:relerrs}
  \centering
  \begin{tabular}{@{}l@{\enspace}l@{\enspace}r@{\enspace}r@{\enspace}r@{\enspace}r@{\enspace}r@{\enspace}r@{\enspace}r@{\enspace}r@{\enspace}r@{}}
    Dataset & Target $\mu$ & \texttt{Naive} & \texttt{RS} & \texttt{RSP} & \texttt{DEANN} & \texttt{DEANNP} & \texttt{HBE} & \texttt{RSA} & \texttt{SKKD} & \texttt{SKBT}\\\hline
    \textsc{ALOI} & 0.01 & 0.000 & 0.095 & 0.090 & 0.100 & 0.102 & 0.110 & 0.099 & 0.076 & 0.091\\
    \textsc{ALOI} & 0.001 & 0.000 & 0.106 & 0.113 & 0.104 & 0.101 & 0.096 & 0.097 & 0.092 & 0.097\\
    \textsc{ALOI} & 0.0001 & 0.000 & 0.102 & 0.099 & 0.100 & 0.100 & \textit{n/a} & \textit{n/a} & 0.098 & 0.098\\
    \textsc{ALOI} & 0.00001 & 0.000 & 0.072 & 0.102 & 0.092 & 0.094 & \textit{n/a} & \textit{n/a} & 0.099 & 0.098\\
    \textsc{last.fm} & 0.01 & 0.001 & 0.061 & 0.052 & 0.111 & 0.114 & \textit{n/a} & \textit{n/a} & 0.094 & 0.091\\
    \textsc{last.fm} & 0.001 & 0.001 & 0.095 & 0.092 & 0.111 & 0.089 & \textit{n/a} & \textit{n/a} & 0.086 & 0.056\\
    \textsc{last.fm} & 0.0001 & 0.002 & 0.056 & 0.086 & 0.109 & 0.108 & \textit{n/a} & \textit{n/a} & 0.051 & 0.073\\
    \textsc{last.fm} & 0.00001 & 0.004 & 0.093 & 0.088 & 0.092 & 0.096 & \textit{n/a} & \textit{n/a} & 0.105 & 0.161\\
    \textsc{SVHN} & 0.01 & 0.000 & 0.081 & 0.081 & 0.092 & 0.093 & 0.109 & 0.048 & 0.098 & 0.098\\
    \textsc{SVHN} & 0.001 & 0.000 & 0.084 & 0.084 & 0.088 & 0.090 & \textit{n/a} & 0.080 & 0.099 & 0.099\\
    \textsc{SVHN} & 0.0001 & 0.000 & 0.076 & 0.087 & 0.090 & 0.091 & \textit{n/a} & 0.053 & 0.099 & 0.099\\
    \textsc{SVHN} & 0.00001 & 0.000 & 0.090 & 0.098 & 0.089 & 0.091 & \textit{n/a} & \textit{n/a} & 0.099 & 0.099\\
  \end{tabular}
\end{table*}

\begin{figure*}[h!!]
  \centering
  \includegraphics[width=\linewidth]{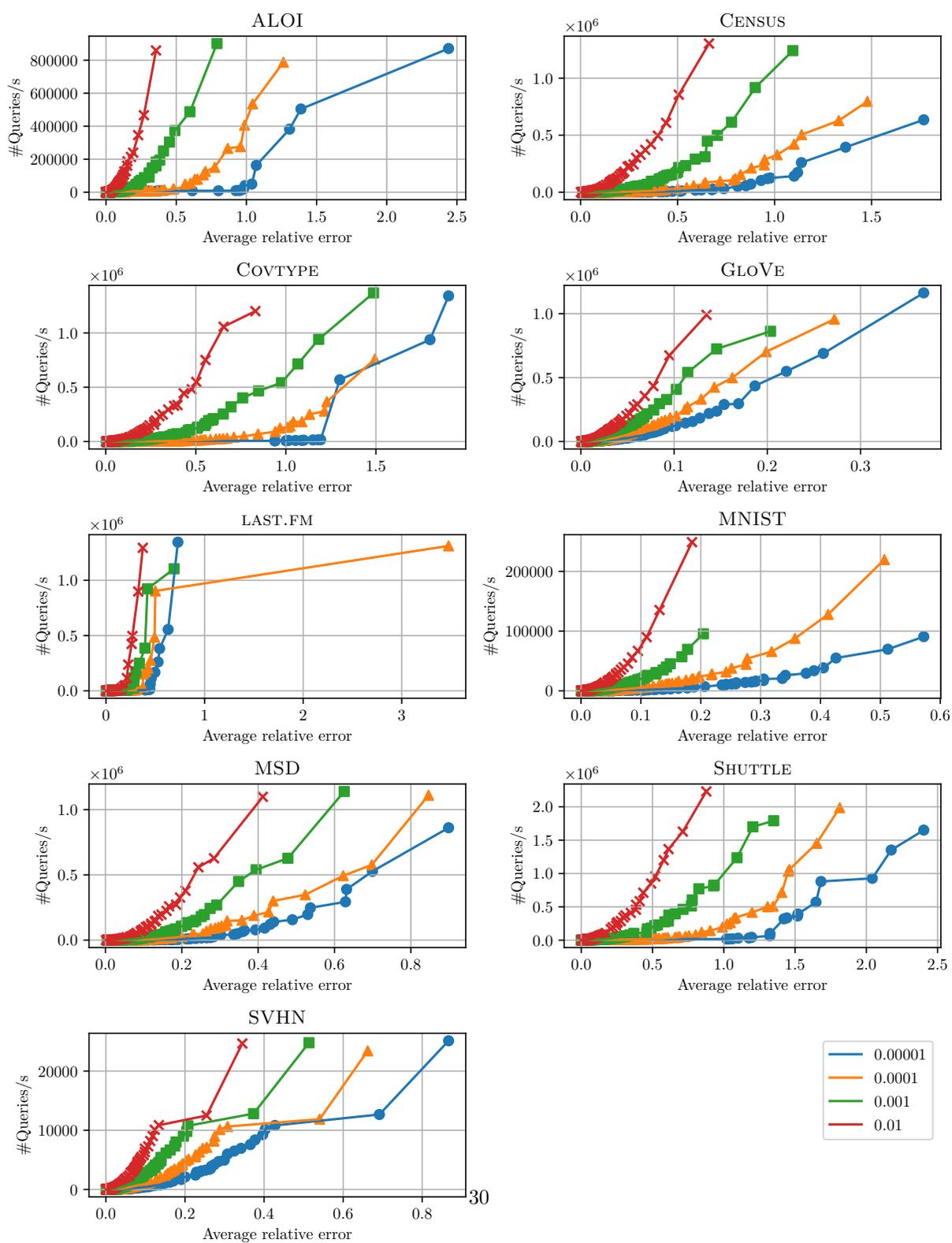}
  \caption{Average Relative Error Vs. Query Time Tradeoff At Different
  Parameter Choices, Reported For \texttt{DEANNP}.}
  \label{fig:errvsqueries}
\end{figure*}

\begin{table*}[h!!]
  \caption{Query Times in Milliseconds / Query and Relative Errors
    Using Fixed Parameters ($k=100$, $m=1000$, $n_\ell=512$, $n_q=1$) vs. Parameters from Grid Search.}
  \label{tbl:fixparamresults}
  \begin{center}
  \begin{tabular}{llrrr}
            &              & Grid Search
    & \multicolumn{2}{c}{Fixed parameters} \\
    \cmidrule(lr){3-3} \cmidrule(lr){4-5}
    Dataset & Target $\mu$ & Query Time & Query Time & Relative Error \\\hline
\textsc{ALOI} & 0.01 & 0.014 & 0.188 & 0.036\\
\textsc{ALOI} & 0.001 & 0.156 & 0.169 & 0.070\\
\textsc{ALOI} & 0.0001 & 0.229 & 0.167 & 0.117\\
\textsc{ALOI} & 0.00001 & 0.209 & 0.167 & 0.174\\
\textsc{Census} & 0.01 & 0.031 & 0.329 & 0.082\\
\textsc{Census} & 0.001 & 0.189 & 0.334 & 0.179\\
\textsc{Census} & 0.0001 & 0.876 & 0.345 & 0.317\\
\textsc{Census} & 0.00001 & 0.889 & 0.331 & 0.452\\
\textsc{Covtype} & 0.01 & 0.046 & 0.177 & 0.099\\
\textsc{Covtype} & 0.001 & 0.272 & 0.176 & 0.206\\
\textsc{Covtype} & 0.0001 & 0.536 & 0.222 & 0.358\\
\textsc{Covtype} & 0.00001 & 0.357 & 0.184 & 0.359\\
\textsc{GloVe} & 0.01 & 0.001 & 0.316 & 0.014\\
\textsc{GloVe} & 0.001 & 0.003 & 0.321 & 0.020\\
\textsc{GloVe} & 0.0001 & 0.005 & 0.316 & 0.029\\
\textsc{GloVe} & 0.00001 & 0.008 & 0.321 & 0.039\\
\textsc{last.fm} & 0.01 & 0.204 & 0.253 & 0.073\\
\textsc{last.fm} & 0.001 & 0.240 & 0.230 & 0.140\\
\textsc{last.fm} & 0.0001 & 0.269 & 0.226 & 0.109\\
\textsc{last.fm} & 0.00001 & 0.286 & 0.230 & 0.214\\
\textsc{MNIST} & 0.01 & 0.015 & 0.492 & 0.018\\
\textsc{MNIST} & 0.001 & 0.052 & 0.494 & 0.034\\
\textsc{MNIST} & 0.0001 & 0.212 & 0.490 & 0.059\\
\textsc{MNIST} & 0.00001 & 0.724 & 0.493 & 0.104\\
\textsc{MSD} & 0.01 & 0.012 & 0.202 & 0.047\\
\textsc{MSD} & 0.001 & 0.054 & 0.201 & 0.083\\
\textsc{MSD} & 0.0001 & 0.266 & 0.199 & 0.126\\
\textsc{MSD} & 0.00001 & 0.426 & 0.208 & 0.175\\
\textsc{Shuttle} & 0.01 & 0.025 & 0.091 & 0.090\\
\textsc{Shuttle} & 0.001 & 0.133 & 0.120 & 0.143\\
\textsc{Shuttle} & 0.0001 & 0.111 & 0.109 & 0.224\\
\textsc{Shuttle} & 0.00001 & 0.078 & 0.101 & 0.274\\
\textsc{SVHN} & 0.01 & 0.145 & 3.133 & 0.032\\
\textsc{SVHN} & 0.001 & 0.423 & 2.925 & 0.052\\
\textsc{SVHN} & 0.0001 & 1.063 & 2.929 & 0.079\\
\textsc{SVHN} & 0.00001 & 2.404 & 2.894 & 0.117\\
  \end{tabular}
  \end{center}
\end{table*}

This section details supplementary experiments that were evaluated
with respect to the Gaussian kernel. The experiments also included
\texttt{ASKIT}~\citep{MarchXB:2015} as a competitor.

The experiments were performed on the same physical hardware as those
with the exponential kernel, but with an improved experimental
framework where all implementations ran inside Docker containers to
isolate them from the rest of the environment; in particular, this
enabled us to run \texttt{ASKIT} which depends on the Intel toolchain,
including the Intel MPI libraries, for compilation. The scripts for
creating the Docker images are included in the code\footnote{\url{https://github.com/mkarppa/deann-experiments}}. In the
experiments, all implementations were limited to 32~GiB of memory, and
any runs that exceeded the memory limitation were terminated. The
total runtime of a run (including preprocessing) was limited to 6
hours, excluding \texttt{SKKD} and \texttt{SKBT}; this means that if
the implementation could not build its datastructures and evaluate the
queries with a certain set of parameters within the time limit, the
run was terminated, and any subsequent runs whose parametrization
would imply a longer runtime were not allowed to run either.

For \texttt{DEANN} and \texttt{DEANNP}, we fixed $n_q = 1$ and let
$n_\ell \in \{ 32, 64, 128, 256, 512, 1024, 2048, 4096\}$. The
perception here is that, for example, the choice $n_q=2,n_\ell=64$ is
essentially the same as $n_q=1,n_\ell=32$ since we would expect to
look at a similar number of near neighbors, assuming the points in the
training set are somewhat well-behaved in their distribution among the
different clusters in the $k$-means that FAISS does in its index
building. The parameters $k$ and $m$ were selected to be from multiple
scales using the formula
\begin{equation}
  \label{eq:kmformula}
  k,m \in \{ 10\cdot(\sqrt{2})^i : i=0,1,\ldots \} \, ,
\end{equation}
such that $k$ and $m$ satisfy $k+m<n$. A cartesian product of the
parameters $(k,m,n_\ell)$ was then probed against
the validation set.

For \texttt{RS} and \texttt{RSP}, the candidates for the parameter $m$ were chosen by the
same formula of Equation~\eqref{eq:kmformula} such that $m < n$.

For \texttt{SKKD} and \texttt{SKBT}, the parameter
$\ell$ was set to $\mathrm{round}(10\cdot(\sqrt{2})^i)$ for
$i\in\{0,1,\ldots,9\}$. The parameter $t_r$ was set to
$\{0,0.05,0.1,0.15,0.2,0.25,0.3,0.35,0.4,0.45,0.5\}$.
A cartesian product of the values $(\ell,t_r)$ was probed against the
validation set.

For \texttt{HBE} and \texttt{RS}, the parameter $\epsilon$ was set to
$\{0.1,0.15,0.2,\ldots, 1.5\}$ and the parameter $\tau$ was set to
$0.01 / (\sqrt{2})^i$ for $i\in\{0,1,\ldots,19\}$ and rounded to 5
decimal digits, so $\tau$ ranged from 0.01 to 0.00001. A cartesian
product of the values $(\epsilon,\tau)$ was probed against the
validation set.

For \texttt{ASKIT}, we followed the suggestions provided by~\cite{MarchXB:2015}.
For the initial pilot experiments,
we set $\text{id}_\text{tol}$ to $10^i$ with $-10 \leq i \leq 2$, 
provide $\kappa = 100$ nearest neighbors for each dataset and query
set point,
set the max number of points $m$ to $2^i$ with $6 \leq i \leq 12$, and set the oversampling factor to 
$f \in \{2, 5, 10\}$. 
After noticing that we could not obtain low relative error for small bandwidth choices by exploring these parameter choices,
we further set skeleton targets $t_\text{skel} \in \{2, 5, 10\}$ and
and set the minimum skeleton level to $\ell_\text{skel} \in \{2, 5, 10, 20\}$.
After these initial experiments, we pruned the parameter space and details can be seen in the script \texttt{generate\_askit.py} in the GitHub repository\footnote{\url{https://github.com/mkarppa/deann-experiments}}.
We remark that \texttt{ASKIT} assumes that the nearest neighbors for
the whole dataset and query set are given as input during preprocessing, which is extremely costly.
On our setup, it took 16 hours using FAISS with multi-threading using 48 threads to precompute this information for the 9 datasets in question.
In contrast, our variants of \texttt{DEANN} compute these neighbors during query time.

When computing the relative error (especially in the validation step),
we excluded points whose exact KDE value was below $10^{-16}$, since
the single-precision floating-point arithmetic turned out to be
numerically too unstable to be useful at that point.

The main results are shown in
Table~\ref{tbl:mainresultsgaussian}. The results agree with the
those obtained with the exponential kernel; in almost all cases,
either \texttt{DEANNP} or \texttt{RSP} is the fastest algorithm, and
\texttt{DEANNP} is seldom very far behind, as it can fall back to
essentially the same random sampling algorithm. Surprisingly, in the
cases of \textsc{last.fm} and \textsc{MNIST} at target KDE value of
$0.00001$, the \texttt{Naive} algorithm turned out to be
unbeatable. Low error requires too many points to be looked at for
approximate algorithms to be very effective in this low bandwidth
regime.

Table~\ref{tbl:mainresultsgaussiankeqm} shows complementary results
where we have limited ourselves to the case that $k=m$ when tuning the
parameters of \texttt{DEANN} and \texttt{DEANNP}. To make the effect
even clearer, Table~\ref{tbl:mainresultsgaussianratio} shows the ratio
of the runtime for \texttt{DEANN} and \texttt{DEANNP} from
Table~\ref{tbl:mainresultsgaussiankeqm} divided by the
runtime with the best parameters from
Table~\ref{tbl:mainresultsgaussian}. While it is clear that sometimes the
number of nearest neighbors and random samples that one ought to look
at are unbalanced, the situation is not desperate, and in many cases
useful results can be obtained even while restricting the search space.

\begin{table*}[h!!]
  \begin{center}
    \caption{Results of Evaluating the Different Algorithms Against
      the Test Set in Milliseconds / Query with the Gaussian Kernel.}
    \label{tbl:mainresultsgaussian}
\begin{tabular}{@{}l@{\,}l@{\,}r@{\,}r@{\,}r@{\,}r@{\,}r@{\,}r@{\,}r@{\,}r@{\,}r@{\,}r@{}}
Dataset & Target $\mu$ & \texttt{Naive} & \texttt{RS} & \texttt{RSP} & \texttt{DEANN} & \texttt{DEANNP} & \texttt{HBE} & \texttt{RSA} & \texttt{SKKD} & \texttt{SKBT} & \texttt{ASKIT}\\\hline\textsc{ALOI} & 0.01 & 1.036 & 0.329 & \textbf{0.084} & 0.286 & 0.089 & 4.748 & 15.502 & 51.882 & 43.581 & 0.491\\
\textsc{ALOI} & 0.001 & 1.002 & 7.858 & 1.676 & 0.768 & \textbf{0.296} & \textit{n/a} & 971.296 & 48.001 & 44.614 & 1.242\\
\textsc{ALOI} & 0.0001 & 1.079 & 16.646 & 5.128 & 1.552 & \textbf{0.519} & \textit{n/a} & \textit{n/a} & 36.452 & 41.628 & 1.252\\
\textsc{ALOI} & 0.00001 & 1.894 & \textit{n/a} & \textit{n/a} & 0.515 & \textbf{0.446} & \textit{n/a} & \textit{n/a} & 24.127 & 31.434 & 1.236\\
\textsc{Census} & 0.01 & 21.167 & 0.718 & \textbf{0.117} & 0.749 & 0.207 & 2.056 & 69.093 & 286.878 & 375.744 & \textit{n/a}\\
\textsc{Census} & 0.001 & 22.291 & 6.963 & 1.239 & 1.550 & \textbf{0.891} & \textit{n/a} & \textit{n/a} & 160.692 & 353.577 & \textit{n/a}\\
\textsc{Census} & 0.0001 & 23.465 & 65.938 & 15.354 & 2.350 & \textbf{1.400} & \textit{n/a} & \textit{n/a} & 116.222 & 295.493 & \textit{n/a}\\
\textsc{Census} & 0.00001 & 28.710 & 278.683 & 62.955 & 2.520 & \textbf{1.582} & \textit{n/a} & \textit{n/a} & 88.905 & 249.631 & \textit{n/a}\\
\textsc{Covtype} & 0.01 & 6.305 & 0.612 & \textbf{0.091} & 0.637 & 0.128 & 0.708 & 43.133 & 23.449 & 24.650 & 3.334\\
\textsc{Covtype} & 0.001 & 6.467 & 4.844 & 0.808 & 1.294 & \textbf{0.672} & 6.032 & \textit{n/a} & 8.982 & 10.943 & 4.970\\
\textsc{Covtype} & 0.0001 & 5.984 & 47.721 & 6.112 & 1.423 & \textbf{1.079} & \textit{n/a} & \textit{n/a} & 2.900 & 4.348 & 5.289\\
\textsc{Covtype} & 0.00001 & 5.027 & \textit{n/a} & \textit{n/a} & 0.242 & \textbf{0.225} & \textit{n/a} & \textit{n/a} & 0.628 & 1.171 & 5.495\\
\textsc{GloVe} & 0.01 & 10.650 & 0.021 & \textbf{0.004} & 0.021 & 0.005 & 8.908 & 0.867 & 577.850 & 549.691 & \textit{n/a}\\
\textsc{GloVe} & 0.001 & 10.675 & 0.071 & 0.021 & 0.065 & \textbf{0.013} & \textit{n/a} & 1.769 & 574.288 & 551.816 & 1.423\\
\textsc{GloVe} & 0.0001 & 10.447 & 0.163 & \textbf{0.042} & 0.162 & 0.069 & \textit{n/a} & 38.040 & 578.709 & 551.685 & 3.548\\
\textsc{GloVe} & 0.00001 & 10.320 & 0.880 & 0.189 & 0.429 & \textbf{0.167} & \textit{n/a} & 1862.553 & 630.915 & 499.688 & \textit{n/a}\\
\textsc{lastfm} & 0.01 & 2.989 & 32.091 & 5.142 & 0.583 & \textbf{0.321} & \textit{n/a} & \textit{n/a} & 90.750 & 78.711 & \textit{n/a}\\
\textsc{lastfm} & 0.001 & 3.022 & 51.482 & 7.119 & 1.018 & \textbf{0.496} & \textit{n/a} & \textit{n/a} & 83.927 & 86.719 & \textit{n/a}\\
\textsc{lastfm} & 0.0001 & 3.046 & 43.680 & 7.196 & 2.683 & \textbf{0.848} & \textit{n/a} & \textit{n/a} & 89.954 & 87.373 & \textit{n/a}\\
\textsc{lastfm} & 0.00001 & \textbf{2.888} & \textit{n/a} & 7.133 & 7.260 & 3.841 & \textit{n/a} & \textit{n/a} & 88.133 & \textit{n/a} & \textit{n/a}\\
\textsc{MNIST} & 0.01 & 1.477 & 0.132 & 0.087 & 0.105 & \textbf{0.067} & 16.628 & 8.517 & 95.594 & 63.673 & 0.684\\
\textsc{MNIST} & 0.001 & 1.466 & 0.939 & 0.602 & 0.679 & \textbf{0.382} & \textit{n/a} & 146.002 & 90.962 & 59.170 & 3.058\\
\textsc{MNIST} & 0.0001 & 1.594 & 5.409 & 3.385 & 1.655 & \textbf{1.085} & \textit{n/a} & 19838.656 & 96.430 & 63.810 & 2.984\\
\textsc{MNIST} & 0.00001 & \textbf{1.457} & 32.994 & 13.151 & 3.743 & 2.602 & \textit{n/a} & \textit{n/a} & 95.879 & 63.010 & 3.226\\
\textsc{MSD} & 0.01 & 4.891 & 3.546 & 0.434 & 0.574 & \textbf{0.256} & \textit{n/a} & \textit{n/a} & 161.008 & 188.234 & 5.108\\
\textsc{MSD} & 0.001 & 5.224 & 36.903 & 6.724 & 1.603 & \textbf{0.663} & \textit{n/a} & \textit{n/a} & 131.380 & 162.818 & 5.251\\
\textsc{MSD} & 0.0001 & 5.585 & 70.699 & 13.940 & 4.316 & \textbf{1.693} & \textit{n/a} & \textit{n/a} & 128.709 & 180.890 & 5.126\\
\textsc{MSD} & 0.00001 & 5.982 & 101.900 & 13.784 & 10.268 & \textbf{3.783} & \textit{n/a} & \textit{n/a} & 115.804 & 175.152 & 5.262\\
\textsc{Shuttle} & 0.01 & 0.519 & 0.368 & \textbf{0.045} & 0.270 & 0.086 & 1.153 & 16.054 & 2.291 & 2.919 & 0.329\\
\textsc{Shuttle} & 0.001 & 0.555 & 2.993 & 0.233 & 0.316 & \textbf{0.175} & 36.222 & \textit{n/a} & 1.297 & 2.631 & 0.331\\
\textsc{Shuttle} & 0.0001 & 0.497 & \textit{n/a} & \textit{n/a} & \textbf{0.084} & 0.087 & \textit{n/a} & \textit{n/a} & 0.622 & 2.230 & 0.363\\
\textsc{Shuttle} & 0.00001 & 0.411 & \textit{n/a} & \textit{n/a} & 0.049 & \textbf{0.047} & \textit{n/a} & \textit{n/a} & 0.284 & 1.972 & 0.407\\
\textsc{SVHN} & 0.01 & 40.710 & 1.129 & \textbf{0.836} & 0.856 & 1.145 & \textit{n/a} & \textit{n/a} & 2797.206 & 1925.226 & \textit{n/a}\\
\textsc{SVHN} & 0.001 & 40.653 & 4.440 & 4.545 & 4.377 & \textbf{2.841} & \textit{n/a} & \textit{n/a} & 2570.958 & 1920.286 & \textit{n/a}\\
\textsc{SVHN} & 0.0001 & 40.241 & 55.776 & 23.587 & 12.627 & \textbf{11.466} & \textit{n/a} & \textit{n/a} & 2548.359 & 1930.954 & \textit{n/a}\\
\textsc{SVHN} & 0.00001 & 41.309 & 279.773 & 140.986 & 52.181 & \textbf{36.585} & \textit{n/a} & \textit{n/a} & 2529.778 & 1869.640 & \textit{n/a}\\
\end{tabular}
  \end{center}
\end{table*}

\begin{table*}[h!!]
  \begin{center}
    \caption{Results of Evaluating the Different Algorithms Against
      the Test Set in Milliseconds / Query with the Gaussian Kernel
      with the Restriction that $k=m$.}
    \label{tbl:mainresultsgaussiankeqm}
\begin{tabular}{@{}l@{\,}l@{\,}r@{\,}r@{\,}r@{\,}r@{\,}r@{\,}r@{\,}r@{\,}r@{\,}r@{\,}r@{}}
Dataset & Target $\mu$ & \texttt{Naive} & \texttt{RS} & \texttt{RSP} & \texttt{DEANN} & \texttt{DEANNP} & \texttt{HBE} & \texttt{RSA} & \texttt{SKKD} & \texttt{SKBT} & \texttt{ASKIT}\\\hline\textsc{ALOI} & 0.01 & 1.036 & 0.329 & \textbf{0.084} & 0.270 & 0.193 & 4.748 & 15.502 & 51.882 & 43.581 & 0.491\\
\textsc{ALOI} & 0.001 & 1.002 & 7.858 & 1.676 & 0.768 & \textbf{0.338} & \textit{n/a} & 971.296 & 48.001 & 44.614 & 1.242\\
\textsc{ALOI} & 0.0001 & \textbf{1.079} & 16.646 & 5.128 & 2.586 & 1.245 & \textit{n/a} & \textit{n/a} & 36.452 & 41.628 & 1.252\\
\textsc{ALOI} & 0.00001 & 1.894 & \textit{n/a} & \textit{n/a} & 7.061 & 2.393 & \textit{n/a} & \textit{n/a} & 24.127 & 31.434 & \textbf{1.236}\\
\textsc{Census} & 0.01 & 21.167 & 0.718 & \textbf{0.117} & 0.778 & 0.561 & 2.056 & 69.093 & 286.878 & 375.744 & \textit{n/a}\\
\textsc{Census} & 0.001 & 22.291 & 6.963 & 1.239 & 2.098 & \textbf{1.192} & \textit{n/a} & \textit{n/a} & 160.692 & 353.577 & \textit{n/a}\\
\textsc{Census} & 0.0001 & 23.465 & 65.938 & 15.354 & 3.503 & \textbf{2.570} & \textit{n/a} & \textit{n/a} & 116.222 & 295.493 & \textit{n/a}\\
\textsc{Census} & 0.00001 & 28.710 & 278.683 & 62.955 & 3.496 & \textbf{2.625} & \textit{n/a} & \textit{n/a} & 88.905 & 249.631 & \textit{n/a}\\
\textsc{Covtype} & 0.01 & 6.305 & 0.612 & \textbf{0.091} & 0.654 & 0.399 & 0.708 & 43.133 & 23.449 & 24.650 & 3.334\\
\textsc{Covtype} & 0.001 & 6.467 & 4.844 & \textbf{0.808} & 1.900 & 0.983 & 6.032 & \textit{n/a} & 8.982 & 10.943 & 4.970\\
\textsc{Covtype} & 0.0001 & 5.984 & 47.721 & 6.112 & 2.443 & \textbf{1.924} & \textit{n/a} & \textit{n/a} & 2.900 & 4.348 & 5.289\\
\textsc{Covtype} & 0.00001 & 5.027 & \textit{n/a} & \textit{n/a} & 0.594 & \textbf{0.477} & \textit{n/a} & \textit{n/a} & 0.628 & 1.171 & 5.495\\
\textsc{GloVe} & 0.01 & 10.650 & 0.021 & \textbf{0.004} & 0.172 & 0.168 & 8.908 & 0.867 & 577.850 & 549.691 & \textit{n/a}\\
\textsc{GloVe} & 0.001 & 10.675 & 0.071 & \textbf{0.021} & 0.283 & 0.234 & \textit{n/a} & 1.769 & 574.288 & 551.816 & 1.423\\
\textsc{GloVe} & 0.0001 & 10.447 & 0.163 & \textbf{0.042} & 0.393 & 0.359 & \textit{n/a} & 38.040 & 578.709 & 551.685 & 3.548\\
\textsc{GloVe} & 0.00001 & 10.320 & 0.880 & \textbf{0.189} & 0.693 & 0.471 & \textit{n/a} & 1862.553 & 630.915 & 499.688 & \textit{n/a}\\
\textsc{lastfm} & 0.01 & 2.989 & 32.091 & 5.142 & 1.077 & \textbf{1.029} & \textit{n/a} & \textit{n/a} & 90.750 & 78.711 & \textit{n/a}\\
\textsc{lastfm} & 0.001 & \textbf{3.022} & 51.482 & 7.119 & 7.126 & 3.992 & \textit{n/a} & \textit{n/a} & 83.927 & 86.719 & \textit{n/a}\\
\textsc{lastfm} & 0.0001 & \textbf{3.046} & 43.680 & 7.196 & 14.758 & 6.060 & \textit{n/a} & \textit{n/a} & 89.954 & 87.373 & \textit{n/a}\\
\textsc{lastfm} & 0.00001 & \textbf{2.888} & \textit{n/a} & 7.133 & 30.810 & 10.142 & \textit{n/a} & \textit{n/a} & 88.133 & \textit{n/a} & \textit{n/a}\\
\textsc{MNIST} & 0.01 & 1.477 & 0.132 & \textbf{0.087} & 0.300 & 0.310 & 16.628 & 8.517 & 95.594 & 63.673 & 0.684\\
\textsc{MNIST} & 0.001 & 1.466 & 0.939 & 0.602 & 0.679 & \textbf{0.572} & \textit{n/a} & 146.002 & 90.962 & 59.170 & 3.058\\
\textsc{MNIST} & 0.0001 & 1.594 & 5.409 & 3.385 & 1.655 & \textbf{1.392} & \textit{n/a} & 19838.656 & 96.430 & 63.810 & 2.984\\
\textsc{MNIST} & 0.00001 & \textbf{1.457} & 32.994 & 13.151 & 4.944 & 3.987 & \textit{n/a} & \textit{n/a} & 95.879 & 63.010 & 3.226\\
\textsc{MSD} & 0.01 & 4.891 & 3.546 & \textbf{0.434} & 1.118 & 0.671 & \textit{n/a} & \textit{n/a} & 161.008 & 188.234 & 5.108\\
\textsc{MSD} & 0.001 & 5.224 & 36.903 & 6.724 & 2.635 & \textbf{2.143} & \textit{n/a} & \textit{n/a} & 131.380 & 162.818 & 5.251\\
\textsc{MSD} & 0.0001 & 5.585 & 70.699 & 13.940 & 13.175 & 6.986 & \textit{n/a} & \textit{n/a} & 128.709 & 180.890 & \textbf{5.126}\\
\textsc{MSD} & 0.00001 & 5.982 & 101.900 & 13.784 & 18.616 & 15.204 & \textit{n/a} & \textit{n/a} & 115.804 & 175.152 & \textbf{5.262}\\
\textsc{Shuttle} & 0.01 & 0.519 & 0.368 & \textbf{0.045} & 0.270 & 0.194 & 1.153 & 16.054 & 2.291 & 2.919 & 0.329\\
\textsc{Shuttle} & 0.001 & 0.555 & 2.993 & \textbf{0.233} & 0.316 & 0.272 & 36.222 & \textit{n/a} & 1.297 & 2.631 & 0.331\\
\textsc{Shuttle} & 0.0001 & 0.497 & \textit{n/a} & \textit{n/a} & \textbf{0.133} & 0.156 & \textit{n/a} & \textit{n/a} & 0.622 & 2.230 & 0.363\\
\textsc{Shuttle} & 0.00001 & 0.411 & \textit{n/a} & \textit{n/a} & \textbf{0.051} & 0.052 & \textit{n/a} & \textit{n/a} & 0.284 & 1.972 & 0.407\\
\textsc{SVHN} & 0.01 & 40.710 & 1.129 & \textbf{0.836} & 4.370 & 3.141 & \textit{n/a} & \textit{n/a} & 2797.206 & 1925.226 & \textit{n/a}\\
\textsc{SVHN} & 0.001 & 40.653 & \textbf{4.440} & 4.545 & 6.220 & 6.680 & \textit{n/a} & \textit{n/a} & 2570.958 & 1920.286 & \textit{n/a}\\
\textsc{SVHN} & 0.0001 & 40.241 & 55.776 & 23.587 & \textbf{13.814} & 14.836 & \textit{n/a} & \textit{n/a} & 2548.359 & 1930.954 & \textit{n/a}\\
\textsc{SVHN} & 0.00001 & \textbf{41.309} & 279.773 & 140.986 & 52.181 & 43.192 & \textit{n/a} & \textit{n/a} & 2529.778 & 1869.640 & \textit{n/a}\\
\end{tabular}
  \end{center}
\end{table*}

\begin{table*}[h!!]
  \begin{center}
    \caption{Runtime Ratio Between Best Parameters and $k=m$ Parameters.}
    \label{tbl:mainresultsgaussianratio}
\begin{tabular}{@{}l@{\,}l@{\,}r@{\,}r@{}}
Dataset & Target $\mu$ & \texttt{DEANN} & \texttt{DEANNP}\\\hline
\textsc{ALOI} & 0.01 & 0.943 & 2.178\\
\textsc{ALOI} & 0.001 & 1.000 & 1.141\\
\textsc{ALOI} & 0.0001 & 1.666 & 2.397\\
\textsc{ALOI} & 0.00001 & 13.707 & 5.366\\
\textsc{Census} & 0.01 & 1.040 & 2.707\\
\textsc{Census} & 0.001 & 1.353 & 1.339\\
\textsc{Census} & 0.0001 & 1.491 & 1.836\\
\textsc{Census} & 0.00001 & 1.387 & 1.659\\
\textsc{Covtype} & 0.01 & 1.028 & 3.122\\
\textsc{Covtype} & 0.001 & 1.468 & 1.463\\
\textsc{Covtype} & 0.0001 & 1.716 & 1.784\\
\textsc{Covtype} & 0.00001 & 2.454 & 2.115\\
\textsc{GloVe} & 0.01 & 8.295 & 33.587\\
\textsc{GloVe} & 0.001 & 4.378 & 18.506\\
\textsc{GloVe} & 0.0001 & 2.430 & 5.221\\
\textsc{GloVe} & 0.00001 & 1.616 & 2.811\\
\textsc{last.fm} & 0.01 & 1.848 & 3.209\\
\textsc{last.fm} & 0.001 & 7.000 & 8.051\\
\textsc{last.fm} & 0.0001 & 5.500 & 7.148\\
\textsc{last.fm} & 0.00001 & 4.244 & 2.640\\
\textsc{MNIST} & 0.01 & 2.861 & 4.647\\
\textsc{MNIST} & 0.001 & 1.000 & 1.500\\
\textsc{MNIST} & 0.0001 & 1.000 & 1.282\\
\textsc{MNIST} & 0.00001 & 1.321 & 1.532\\
\textsc{MSD} & 0.01 & 1.949 & 2.618\\
\textsc{MSD} & 0.001 & 1.644 & 3.234\\
\textsc{MSD} & 0.0001 & 3.052 & 4.126\\
\textsc{MSD} & 0.00001 & 1.813 & 4.019\\
\textsc{Shuttle} & 0.01 & 1.000 & 2.245\\
\textsc{Shuttle} & 0.001 & 1.000 & 1.550\\
\textsc{Shuttle} & 0.0001 & 1.576 & 1.779\\
\textsc{Shuttle} & 0.00001 & 1.044 & 1.106\\
\textsc{SVHN} & 0.01 & 5.103 & 2.743\\
\textsc{SVHN} & 0.001 & 1.421 & 2.351\\
\textsc{SVHN} & 0.0001 & 1.094 & 1.294\\
\textsc{SVHN} & 0.00001 & 1.000 & 1.181\\
\end{tabular}

  \end{center}
\end{table*}

\begin{table*}[h!!]
  \begin{center}
    \caption{Preprocessing Times When Evaluating the Different Algorithms Against
      the Test Set with the Gaussian Kernel in Seconds.}
    \label{tbl:buildtimesfull}
\begin{tabular}{@{}l@{\,}l@{\,}r@{\,}r@{\,}r@{\,}r@{\,}r@{\,}r@{\,}r@{\,}r@{\,}r@{\,}r@{}}
Dataset & Target $\mu$ & \texttt{Naive} & \texttt{RS} & \texttt{RSP} & \texttt{DEANN} & \texttt{DEANNP} & \texttt{HBE} & \texttt{RSA} & \texttt{SKKD} & \texttt{SKBT} & \texttt{ASKIT}\\\hline\textsc{ALOI} & 0.01 & 0.006 & \textbf{0.000} & 0.055 & 0.377 & 8.775 & 22.285 & 0.000 & 4.929 & 5.155 & 21.455\\
\textsc{ALOI} & 0.001 & 0.006 & \textbf{0.000} & 0.059 & 1.078 & 0.975 & \textit{n/a} & 0.000 & 5.349 & 5.680 & 6.626\\
\textsc{ALOI} & 0.0001 & 0.006 & \textbf{0.000} & 0.056 & 0.146 & 0.153 & \textit{n/a} & \textit{n/a} & 5.588 & 5.756 & 6.425\\
\textsc{ALOI} & 0.00001 & \textbf{0.006} & \textit{n/a} & \textit{n/a} & 0.154 & 0.146 & \textit{n/a} & \textit{n/a} & 5.782 & 4.949 & 6.372\\
\textsc{Census} & 0.01 & 0.081 & \textbf{0.000} & 0.945 & 3.568 & 14.269 & 101.727 & 0.000 & 25573.250 & 22917.678 & \textit{n/a}\\
\textsc{Census} & 0.001 & 0.077 & \textbf{0.000} & 0.922 & 14.238 & 39.615 & \textit{n/a} & \textit{n/a} & 25405.120 & 25516.643 & \textit{n/a}\\
\textsc{Census} & 0.0001 & 0.075 & \textbf{0.000} & 0.971 & 3.346 & 6.440 & \textit{n/a} & \textit{n/a} & 25486.341 & 25455.323 & \textit{n/a}\\
\textsc{Census} & 0.00001 & 0.068 & \textbf{0.000} & 0.938 & 2.498 & 3.345 & \textit{n/a} & \textit{n/a} & 23033.750 & 23078.722 & \textit{n/a}\\
\textsc{Covtype} & 0.01 & 0.017 & \textbf{0.000} & 0.179 & 26.056 & 0.336 & 11.008 & 0.000 & 5.644 & 4.098 & 572.824\\
\textsc{Covtype} & 0.001 & 0.017 & \textbf{0.000} & 0.189 & 0.439 & 7.590 & 40.260 & \textit{n/a} & 5.413 & 4.360 & 339.547\\
\textsc{Covtype} & 0.0001 & 0.016 & \textbf{0.000} & 0.186 & 0.336 & 0.340 & \textit{n/a} & \textit{n/a} & 5.443 & 4.576 & 76.198\\
\textsc{Covtype} & 0.00001 & \textbf{0.016} & \textit{n/a} & \textit{n/a} & 0.593 & 0.621 & \textit{n/a} & \textit{n/a} & 5.026 & 4.010 & 75.267\\
\textsc{GloVe} & 0.01 & 0.052 & \textbf{0.000} & 0.553 & 1.067 & 2.521 & 135.074 & 0.000 & 29.145 & 30.301 & \textit{n/a}\\
\textsc{GloVe} & 0.001 & 0.049 & \textbf{0.000} & 0.553 & 13.145 & 2.561 & \textit{n/a} & 0.000 & 28.012 & 30.064 & 175.622\\
\textsc{GloVe} & 0.0001 & 0.044 & \textbf{0.000} & 0.553 & 1.240 & 1.267 & \textit{n/a} & 0.000 & 28.417 & 30.106 & 578.841\\
\textsc{GloVe} & 0.00001 & 0.044 & \textbf{0.000} & 0.542 & 144.898 & 12.779 & \textit{n/a} & 0.000 & 33.952 & 26.619 & \textit{n/a}\\
\textsc{last.fm} & 0.01 & 0.010 & \textbf{0.000} & 0.106 & 2.320 & 7.022 & \textit{n/a} & \textit{n/a} & 4.327 & 3.140 & \textit{n/a}\\
\textsc{last.fm} & 0.001 & 0.010 & \textbf{0.000} & 0.105 & 2.317 & 2.331 & \textit{n/a} & \textit{n/a} & 4.168 & 3.226 & \textit{n/a}\\
\textsc{last.fm} & 0.0001 & 0.009 & \textbf{0.000} & 0.108 & 0.249 & 2.312 & \textit{n/a} & \textit{n/a} & 4.148 & 3.193 & \textit{n/a}\\
\textsc{last.fm} & 0.00001 & \textbf{0.009} & \textit{n/a} & 0.107 & 0.794 & 0.850 & \textit{n/a} & \textit{n/a} & 4.125 & \textit{n/a} & \textit{n/a}\\
\textsc{MNIST} & 0.01 & 0.017 & \textbf{0.000} & 0.159 & 1.700 & 0.813 & 100.323 & 0.000 & 12.369 & 11.154 & 14.053\\
\textsc{MNIST} & 0.001 & 0.017 & \textbf{0.000} & 0.156 & 1.739 & 1.650 & \textit{n/a} & 0.000 & 12.051 & 10.671 & 4.424\\
\textsc{MNIST} & 0.0001 & 0.016 & \textbf{0.000} & 0.163 & 0.838 & 1.768 & \textit{n/a} & 0.000 & 12.097 & 11.201 & 4.423\\
\textsc{MNIST} & 0.00001 & 0.016 & \textbf{0.000} & 0.155 & 0.447 & 0.443 & \textit{n/a} & \textit{n/a} & 12.461 & 11.022 & 4.397\\
\textsc{MSD} & 0.01 & 0.020 & \textbf{0.000} & 0.223 & 9.319 & 9.395 & \textit{n/a} & \textit{n/a} & 12.378 & 10.359 & 144.805\\
\textsc{MSD} & 0.001 & 0.021 & \textbf{0.000} & 0.227 & 0.443 & 0.785 & \textit{n/a} & \textit{n/a} & 10.183 & 9.093 & 143.934\\
\textsc{MSD} & 0.0001 & 0.019 & \textbf{0.000} & 0.223 & 0.432 & 0.435 & \textit{n/a} & \textit{n/a} & 11.974 & 10.106 & 145.066\\
\textsc{MSD} & 0.00001 & 0.019 & \textbf{0.000} & 0.224 & 0.446 & 0.460 & \textit{n/a} & \textit{n/a} & 11.614 & 10.028 & 144.940\\
\textsc{Shuttle} & 0.01 & 0.001 & \textbf{0.000} & 0.007 & 0.238 & 0.070 & 2.006 & 0.000 & 0.687 & 0.658 & 0.593\\
\textsc{Shuttle} & 0.001 & 0.001 & \textbf{0.000} & 0.007 & 0.049 & 0.067 & 63.746 & \textit{n/a} & 0.659 & 0.621 & 1.634\\
\textsc{Shuttle} & 0.0001 & \textbf{0.001} & \textit{n/a} & \textit{n/a} & 0.046 & 0.048 & \textit{n/a} & \textit{n/a} & 0.670 & 0.639 & 1.679\\
\textsc{Shuttle} & 0.00001 & \textbf{0.001} & \textit{n/a} & \textit{n/a} & 0.262 & 0.257 & \textit{n/a} & \textit{n/a} & 0.636 & 0.619 & 1.542\\
\textsc{SVHN} & 0.01 & 0.583 & \textbf{0.000} & 5.590 & 262.613 & 1651.727 & \textit{n/a} & \textit{n/a} & 454.764 & 473.117 & \textit{n/a}\\
\textsc{SVHN} & 0.001 & 0.582 & \textbf{0.000} & 5.624 & 35.396 & 13.996 & \textit{n/a} & \textit{n/a} & 446.738 & 462.035 & \textit{n/a}\\
\textsc{SVHN} & 0.0001 & 0.657 & \textbf{0.000} & 5.662 & 36.748 & 35.219 & \textit{n/a} & \textit{n/a} & 445.377 & 463.516 & \textit{n/a}\\
\textsc{SVHN} & 0.00001 & 0.772 & \textbf{0.000} & 5.592 & 16.252 & 16.640 & \textit{n/a} & \textit{n/a} & 431.374 & 452.096 & \textit{n/a}\\
\end{tabular}
  \end{center}
\end{table*}

\section{Preprocessing Times}
\label{app:preprocessingtimes}

Table~\ref{tbl:buildtimesfull} lists the preprocessing times for the
various algorithms. The reported times are in seconds and were
obtained when evaluating the test using the Gaussian
kernel.

Generally, \texttt{RS} has the smallest construction time,
which is almost zero, since there is no data structure to construct;
the code only stores a pointer to the data. The larger (but still
insignificant) construction time for \texttt{Naive} is explained by
the fact that, upon construction, the data structure stores the
Euclidean norm of each vector in the training set. The construction
time for \texttt{RSP} consists of copying the training set into memory
in random order.

For \texttt{DEANN} and \texttt{DEANNP}, there is a huge variance among
the construction times. This is explained by the fact that it is
dominated by the construction time for FAISS, and is very sensitive to
the parameter $n_\ell$, the number of clusters. Thus the construction
time can range from almost-insignificant to rather long, depending on
the construction time of the underlying NN object. Like \texttt{RS},
\texttt{DEANN} has no construction time of its own, whereas
\texttt{DEANNP} performs the same initialization of random sampling as
\texttt{RSP}.

In all instances where a comparison could be made, \texttt{DEANNP} and
\texttt{DEANN} can be constructed considerably faster than
\texttt{HBE} or \texttt{ASKIT}. \texttt{SKKD} and \texttt{SKBT} are in
a class of their own with respect to construction times, reaching over
7 hours for the \textsc{Census} dataset.
\end{document}